\documentclass{amsart}

\usepackage[T1]{fontenc}   
\usepackage[utf8]{inputenc}
\usepackage[english]{babel} 
\usepackage{amsmath,amsfonts,amssymb,mathtools}
\usepackage{amsthm}
\usepackage[a4paper]{geometry}
\usepackage{multicol}
\usepackage{mathrsfs}
\usepackage{url}
\usepackage{xcolor}
\usepackage{stmaryrd}
\usepackage{algorithm}
 \usepackage[noend]{algpseudocode}
 \usepackage{tikz}
\usepackage{tikz-cd}
\usetikzlibrary{tikzmark,calc,arrows, shapes, decorations.pathreplacing, positioning, quotes}
\usepackage{breakcites}
\usepackage{thmtools}
\usepackage{thm-restate}
\usepackage[colorlinks=true, linkcolor=blue, citecolor=blue, urlcolor=blue]{hyperref}
\usepackage{orcidlink}

\newtheorem{theorem}{Theorem}
{\bfseries}{\itshape}
{\bfseries}{\itshape}

\newtheorem{proposition}{Proposition}
\newtheorem{corollary}{Corollary}
\newtheorem{heur}{Heuristic}
\newtheorem{lemma}{Lemma}

\newtheorem{problem}{Problem}

\newtheorem{definition}{Definition}

\theoremstyle{remark}
\newtheorem{remark}{Remark}

\definecolor{darkgreen}{rgb}{0.0, 0.5, 0.0} 

\newcommand{\Group}{\mathscr{G}(\Fqm/\Fq;r)}

\newcommand{\ds}{\displaystyle}

\newcommand{\eval}{\mathrm{ev}}

\newcommand{\F}{\ensuremath{\mathbb{F}}}
\newcommand{\Fq}{\ensuremath{\mathbb{F}_q}}

\newcommand{\Fqm}{\ensuremath{\mathbb{F}_{q^m}}}

\newcommand{\Gal}[1]{\mathrm{Gal}(#1)}
\newcommand{\Mat}{\mathrm{Mat}}
\newcommand{\Mata}{\mathrm{Mat}_{\alpha}}

\newcommand{\Res}[2]{\mathrm{Res}_{#1}(#2)}
\newcommand{\Resq}[1]{\mathrm{Res}_{\Fqm/\Fq}(#1)}

\newcommand{\word}[1]{\ensuremath{\boldsymbol{#1}}}
\newcommand{\mat}[1]{\ensuremath{\boldsymbol{#1}}}
\newcommand{\code}[1]{\ensuremath{\mathscr{#1}}}

\newcommand{\AC}{\code{A}}

\newcommand{\CC}{\code{C}}
\newcommand{\DC}{\code{D}}

\newcommand{\GC}{\code{G}}

\newcommand{\Am}{\mat{A}}
\newcommand{\Bm}{\mat{B}}

\newcommand{\Gm}{\mat{G}}
\newcommand{\Hm}{\mat{H}}
\renewcommand{\Im}{\mat{I}}
\newcommand{\Jm}{\mat{J}}
\newcommand{\Mm}{\mat{M}} 
        
\newcommand{\Pm}{\mat{P}}
\newcommand{\Qm}{\mat{Q}}
\newcommand{\Rm}{\mat{R}}
\newcommand{\Sm}{\mat{S}}

\newcommand{\Vm}{\mat{V}}

\newcommand{\Thetam}{\mat{\Theta}}
\newcommand{\Gpub}{\Gm_{\mathrm{pub}}}
\newcommand{\Hpub}{\Hm_{\mathrm{pub}}}
\newcommand{\Hsec}{\Hm_{\mathrm{sec}}}
\newcommand{\Vpub}{V_{\mathrm{pub}}}
\newcommand{\Vsec}{V_{\mathrm{sec}}}

\newcommand{\cv}{\mat{c}}

\newcommand{\ev}{\mat{e}}
\newcommand{\gv}{\mat{g}}

\newcommand{\rv}{\mat{r}}
\newcommand{\uv}{\mat{u}}
\newcommand{\vv}{\mat{v}}
\newcommand{\wv}{\mat{w}}
\newcommand{\xv}{{\mat{x}}}
\newcommand{\yv}{{\mat{y}}}

\newcommand{\cA}{\mathcal{A}}
\newcommand{\cB}{\mathcal{B}}

\newcommand{\cF}{\mathcal{F}}

\newcommand{\cL}{\mathcal{L}}
\newcommand{\cP}{\mathcal{P}}

\newcommand{\cX}{\mathcal{X}}
\newcommand{\cY}{\mathcal{Y}}

\newcommand{\GRS}[3]{\text{\bf GRS}_{#1}\left(#2,#3\right)}
\newcommand{\Alt}[3]{\code{A}_{#1}\left(#2, #3\right)}
\newcommand{\Goppa}[2]{\code{G}(#1, #2)}

\newcommand{\Tr}[1]{\mathrm{Tr}\left( #1 \right)}

\newcommand{\St}[1]{\mathrm{St}(#1)}

\newcommand{\Cmat}{\code{C}_\mathrm{mat}}

\newcommand{\PP}{\mathbb{P}}

\newcommand{\GL}{\mathbf{GL}}

\newcommand{\eqdef}{\stackrel{\mathrm{def}}{=}}

\newcommand{\ea}{{e_{\AC}}}
\newcommand{\eg}{{e_{\GC}}}

\title{The tangent space attack}
\author{Axel Lemoine \orcidlink{0009-0008-2440-8568}}
\address{Inria Paris, France}
\address{DGA, France}
\email{axel.lemoine@inria.fr}
\keywords{McEliece scheme, Alternant codes, Algebraic-geometry codes, Weil restriction}

\begin{document}
\maketitle
\begin{abstract}
    We propose a new method for retrieving the algebraic structure of a generic alternant code given an arbitrary generator matrix,
    provided certain conditions are met. We then discuss how this challenges the security of the McEliece cryptosystem instantiated with this family of codes.
    The central object of our work is the quadratic hull related to a linear code, defined as the
    intersection of all quadrics passing through the columns of a given generator or parity-check matrix, where the columns are considered
    as points in the affine or projective space. The geometric properties of this object reveal important information about the internal
    algebraic structure of the code. This is particularly evident in the case of generalized Reed-Solomon codes, whose quadratic hull 
    is deeply linked to a well-known algebraic variety called \textit{the rational normal curve}. By utilizing the concept of Weil restriction of affine varieties,
    we demonstrate that the quadratic hull of a generic dual alternant code inherits many interesting features from the
    rational normal curve, on account of the fact that alternant codes are subfield-subcodes of generalized Reed-Solomon codes. If the rate of the generic alternant code is sufficiently high, this allows us to construct a
    polynomial-time algorithm for retrieving the underlying generalized Reed-Solomon code from which the alternant code is defined, which leads to an efficient
    key-recovery attack against the McEliece cryptosystem when instantiated with this class of codes. Finally, we discuss the generalization of this approach to Algebraic-Geometry codes and Goppa codes.
\end{abstract}
\section*{Introduction}
\subsection*{McEliece cryptosystem} The problem of decoding random linear codes, also known as \textit{random decoding problem}, is widely regarded as a difficult problem. It was shown in \cite{BMT78} that this problem is NP-hard, which may be 
thought of as \textit{worst-case} hardness. More interestingly, the random decoding problem has also been deeply studied in the \textit{average case}. After decades of research, the best generic decoding algorithm \cite{BM17} remains exponential,
which makes the decoding problem a good candidate for asymmetric cryptography. Moreover, it is generally agreed that the decoding problem is quantum resistant, given that the best known quantum decoding algorithm \cite{KT17a} also has exponential complexity.
Until very recently, the NIST was still considering several code-based cryptosystems \cite{AABBBBDGGGGMPRSTVZ22,HQC17,ABCCGLMMMNPPPSSSTW20} in the fourth round of its post-quantum standardization competition, whose \cite{HQC17} was eventually declared the winner.

The first code-based public-key cryptosystem was proposed by McEliece back in 1978 \cite{M78}. The idea is to pick a linear code among a family of codes for which efficient 
decoding algorithms exist, provided that some secret structure about the code is known. The public key consists of a generator matrix of the code, which appears random, while 
the private key is represented by an efficient decoding algorithm. To encrypt a message, the sender first encodes it using the public generator matrix, and then deliberately adds
a random error vector whose Hamming weight equals the code's decoding capability. Only the owner of the private key can recover the original message using the efficient decoding algorithm.

The security of the scheme strongly relies on the choice of the family of codes, as the latter has to behave like random codes. Under the hypothesis of \textit{indistinguishability} between the family of codes 
that is being used and random codes, a rigorous security proof was given in \cite{CFS01}. Under such an assumption, breaking the cryptosystem boils down to decoding a random linear code, whose computational intractability has already been discussed.
Many families have been considered for this purpose, for instance generalized Reed-Solomon (GRS) codes by Niederreiter in \cite{N86}, which were proven insecure for such a use case in \cite{SS92} by Sidelnikov and Shestakov, or algebraic-geometry (AG) codes proposed by Janwa and Moreno in \cite{JM96} and 
attacked by Couvreur, M\'arquez-Corbella and Pellikaan in \cite{CMP14}. Interestingly, the family of binary Goppa codes, originally proposed by McEliece, still looks like a relevant proposition nowadays. Despite the existence of distinguishers against Goppa codes, such as 
those presented in \cite{FGOPT11,CMT23a} which both handle high-rate Goppa codes in polynomial time, or \cite{R24} which handles constant-rate Goppa codes in sub-exponential --- but still high --- time, no efficient key-recovery or message-recovery attacks exist 
against McEliece scheme with \textit{binary} Goppa codes. Most notably, the key recovery attack proposed in \cite{BMT23} could only work for generic alternant codes, while \cite{CMT23a} treats the case of codes over a large alphabet and thus excludes the binary case.

\subsection*{Codes and geometry} A generator matrix of a linear code may be thought of as a collection of points in the affine or projective space, each column of the matrix being seen as a point. The \textit{quadratic hull} of a code with respect to such a matrix is defined 
as the intersection of all quadrics that pass through each of these points. For generalized Reed-Solomon codes whose dimension does not exceed half the length, the quadratic hull turns out to be a rational normal curve. Using Weil restriction and the strong rigidity of the objects that
are derived from it, we wish to provide the same kind of results for generic alternant codes.

\subsection*{Our contribution} We propose a new approach for recovering the private structure of a generic alternant code using both algebraic-geometry and Galois theory. Our algorithm has polynomial complexity with respect to the length and dimension of the code,
and is able to fully recover the structure of the code provided that the rate of the alternant code is high enough for the square of its dual not to fill the entire ambient space. To achieve this goal, we compute the 
quadratic hull with respect to the public parity-check matrix of the code. Under some hypotheses, this algebraic variety is the Weil restriction of the affine cone over the rational normal curve, up to an unknown change of basis.
Even after this change of basis, this variety keeps an interesting property: its tangent spaces are stabilized by a certain linear operator. One may therefore compute the space of linear operators that stabilize all
these tangent spaces, which turns out to have a structure of algebra isomorphic to the extension field over which the underlying GRS code is defined. This enables us to compute another parity-check matrix of the code which directly yields a generator matrix of the underlying GRS code, up to conjugation by the Frobenius automorphism. 
The final step is to apply \cite{SS92} to recover a support and a multiplier, and thus an efficient decoding algorithm. We eventually study the application of our framework to Goppa codes and algebraic-geometry codes, and investigate some natural generalization of generic alternant codes that are vulnerable to our attack.

\subsection*{Outline of the paper} The first section is dedicated to the basic notions of coding theory that we will need throughout the paper. In Section \ref{section:alternants} we set up our key-recovery problem, and explain why this naturally leads to the concept of Weil restriction of ideals and affine algebraic varieties.
Section \ref{section:weil} introduces the notion of Weil restriction of scalars, and establishes several results such as Lemma \ref{lemma:equivalence} which identifies subspaces having the structure of a Weil restriction, 
or Theorem \ref{thm:group} which explicitly gives the family of linear automorphisms that preserve this structure. In this section, we also introduce the notion of Weil-properness, that essentially tells us when this framework can be used to analyze the structure of trace codes. Section \ref{section:attack} provides a polynomial-time algorithm for retrieving an efficient decoding algorithm 
for the subfield-subcode being used in the McEliece cryptosystem, be it a generic alternant code, the subfield subcode of a generic AG code, or even a Goppa code of sufficiently low degree, but provided that the square of the dual code is not the entire ambient space. \section{Notations and prerequisites} \label{section:prerequisites}
\subsection{Basic notation}
~\par
~\par
\paragraph{\textit{Fields.}} Throughout the paper, we work with an extension of finite fields $\Fqm/\Fq$, together with a primitive element $\alpha\in\Fqm$, so that $\Fqm=\Fq[\alpha]$. We may sometimes also refer to a generic field as $\F$.
The algebraic closure of $\F$ is denoted by $\overline{\F}$.\\

\paragraph{\textit{Vectors and matrices.}} Codewords, \textit{i.e.} elements of a linear code, are denoted using bold lowercase letters. We use the row-vectors convention for codewords. For geometric points 
in the affine space, we instead use the column-vectors convention. Matrices are denoted with bold capital letters. If $(\lambda_1,\ldots,\lambda_n)\in\F^n$, the notation $\mathrm{Diag}(\lambda_1,\ldots,\lambda_n)$ 
stands for the diagonal matrix with entries given by the $\lambda_i$'s. We freely use this notation to build block-diagonal matrices as well. Finally, $\GL_{k}(\F)$ denotes the general linear group of size $k$ over $\F$,
while $\GL(\F^k)$ denotes the group of linear automorphisms of $\F^k$. We have of course $\GL_k(\F)\simeq\GL(\F^k)$. \\

\paragraph{\textit{Spans and ideals.}} If $E$ is any subset of an $\F$-vector space, we denote by $\mathrm{Span}_{\F}(E)$ the $\F$-vector subspace spanned by $E$. Similarly, when $A$ is a subset of a ring $\Rm$,
the notation $\langle A\rangle$ stands for the ideal generated by $A$.\\

\paragraph{\textit{Partial derivatives.}} Let $\Rm=\F[x_s~|~s\in S]$ be a polynomial ring whose variables are indexed by a finite set $S$. For $f\in\Rm$, we write $\partial_s f$ for the partial derivative of $f$ with respect to $x_s$.\\

\paragraph{\textit{Varieties}} In this paper, the term \textit{algebraic variety} refers to algebraic subsets of $\F^r$ for some integer $r$ and field $\F$, \textit{i.e.} the zero locus associated to an ideal $I\subset\F[X_1,\ldots,X_r]$ in $\F^r$. When we are given such an ideal $I$, we 
denote by $V(I)$ its variety, or sometimes $V_\F(I)$ when we want to emphasize that the ground field is $\F$. If $I$ is the defining ideal of the variety $V$, the coordinate ring of $V$ is denoted with $\F[V]$ and defined by $\F[X_0,\ldots,X_{r-1}]/I$.

\subsection{Linear codes} A linear subspace $\CC\subseteq\F^n$ of dimension $r$ is called an $[n,k]_{\F}$-linear code. We may talk about $[n,k]_q$-codes when $\F=\Fq$. A \textit{generator matrix} of $\CC$ is a matrix $\Gm$ with coefficients in $\F$ whose row space equals $\CC$.
An $[n,k]_{\F}$-code $\CC$ may as well be defined by a \textit{parity-check matrix} $\Hm\in\F^{(n-k)\times n}$ which is such that 
$$\CC=\{\word{x}\in\F^n~|~\Hm \word{x}^\top = 0\}.$$
As $\F^n$ is endowed with the canonical inner product defined by 
\[\forall\word{x},\word{y}\in\F^n,~\word{x}\cdot\word{y}\eqdef\sum_{i=1}^n x_iy_i,\]
we may define the \textit{dual} of a code $\CC$ as $\CC^\bot\eqdef\{\word{y}\in\F^n~|~\forall\word{x}\in\CC,~\word{x}\cdot\word{y}=0\}$. Note that $\Hm\in\F^{(n-k)\times n}$ is a parity-check matrix of $\CC$ if and only if it is a generator
matrix of $\CC^\bot$.

\subsection{Componentwise product of codes} The space $\F^n$ naturally comes with the canonical product algebra structure. The multiplication law will be denoted with the star notation as follows:
$$\forall \word{x},\word{y}\in\F^n,~\word{x}\star\word{y}\eqdef (x_1y_1,\ldots,x_ny_n).$$
This immediately gives rise to the notion of product of codes.
\begin{definition}
For any $\F$-linear codes $\CC,\DC\subset\F^n$, we define
$$\CC\star\DC\eqdef\mathrm{Span}_{\F}\{\word{c}\star\word{d}~|~(\word{c},\word{d})\in\CC\times\DC\}.$$
We also write $\CC^{\star 2}=\CC\star\CC$.
\end{definition}
Note that the componentwise product of codes is associative and may therefore be iterated. We will thus write $\CC^{\star d}$ for $\underbrace{\CC\star\ldots\star\CC}_{d\text{ times}}$.\\
Let $\Gm\in\F^{r\times n}$ be a generator matrix of $\CC$, and let $\gv_j$ denote the $j$-th column of $\Gm$ for all $1\leq j\leq n$. Define $\Rm=\F[X_0,\ldots,X_{r-1}]$ together with the natural evaluation map
\begin{equation}\label{eq:eval}
    \eval_{\Gm}:\begin{cases}
    \Rm &\longrightarrow \F^n\\
    f &\longmapsto (f(\gv_1),\ldots,f(\gv_n)).
    \end{cases}
\end{equation}
Alternatively, following the convention of \cite{R24}, $\eval_{\Gm}$ can be defined as the only homomorphism of (graded) $\F$-algebras mapping each variable $X_i$ onto the $i$-th row $\rv_i$ of $\Gm$ --- here we have $0\leq i<r$ in order to be consistent with the way we enumerate the variables of $\Rm$. The image of this map, which does not depend on the choice of $\Gm$, was denoted in \cite{R24} by $\bigoplus_{d\geq 0}\CC^{\star d}$ and referred to as the \textit{homogeneous coordinate ring} of $\CC$. On the contrary, the kernel of $\eval_{\Gm}$ formally depends on the choice of $\Gm$ and will be denoted by $I(\Gm)$ in this paper. It has a structure of graded ideal, \textit{i.e.}
$$I(\Gm)=\ds\bigoplus_{d\geq 0}I_d(\Gm),$$
where $I_d(\Gm)$ is the homogeneous component of degree $d$ of $I(\Gm)$, that we will refer to as the vanishing ideal of $\Gm$ at degree $d$.
\begin{definition}[Quadratic hull \cite{R20}]
We define the algebraic quadratic hull of $\CC$ with respect to $\Gm$ as the polynomial ideal generated by $I_2(\Gm)$. The algebraic variety induced by the algebraic quadratic hull will be denoted by $V_2(\Gm)$ and referred to as the geometric quadratic hull of $\CC$ with respect to $\Gm$.
\end{definition}
\begin{remark}
Hilbert's Nullstellensatz establishes a correspondence between ideals of $\overline{\F}[X_1,\ldots,X_r]$ and algebraic subsets of $\overline{\F}^r$, meaning that we can work with either the algebraic or the geometric quadratic hull equivalently. However, most of the time we only have access to the $\F$-rational points of the variety, preventing the Nullstellensatz to hold as $\F$ will always be a finite field in this paper. This will require us to be cautious and to always specify which version of 
the quadratic hull --- algebraic or geometric --- we will be working with.
\end{remark}
Note that the vector space $I_2(\Gm)$ is the very same object as the \textit{code of quadratic relations} introduced in \cite{CMT23a}, as we have 
\[ I_2(\Gm)=\left\{\sum_{i\leq j}c_{i,j}X_iX_j~|~\sum_{i\leq j}c_{i,j}\rv_i\star\rv_j=0\right\}.\]
Recall that the dimension of $I_2(\Gm)$ is related to that of $\CC^{\star 2}$ by 
\begin{equation}\label{eq:dim_I2}
\dim I_2(\Gm)=\binom{r+1}{2}-\dim\CC^{\star 2},
\end{equation}
which can be obtained by applying the rank-nullity theorem on the evaluation map $\mathrm{ev}_{\Gm}$ restricted to $\Rm_2$. Although we emphasize the dependence of $I_2(\Gm)$ on the choice of generator matrix, many properties of the algebraic or geometric quadratic hull are actually intrinsic. More precisely, the quadratic hull --- be it algebraic or geometric --- depends only on the code, at least up to some linear transformation.
\begin{proposition}\label{prop:intrinsic}
    Let $\Gm_1,\Gm_2$ be two $r\times n$ generator matrices of an $\F$-linear code $\CC$. Denote by $\Pm\in\GL_r(\F)$ the transition matrix so that $\Gm_2=\Pm\cdot\Gm_1$. Then 
    \begin{itemize}
        \item[$(i)$] $I_2(\Gm_1)=\left\{f^{\Pm}~|~f\in I_2(\Gm_2)\right\}$, where $f^{\Pm}=f((X_1~\ldots~X_r)\cdot\Pm^\top)$;
        \item[$(ii)$] $V_2(\Gm_2)=\left\{\Pm\cdot\vv~|~\vv\in V_2(\Gm_1)\right\}.$
    \end{itemize}
\end{proposition}
\begin{proof}
Denote by $\gv_1,\ldots,\gv_n$ the columns of $\Gm_1$. As the columns of $\Gm_2$ are $\Pm\cdot\gv_1,\ldots,\Pm\cdot\gv_n$, we see that for any quadratic form $f\in\Rm_2$, 
\begin{align*}
    f\in I_2(\Gm_2) &\iff \forall 1\leq i\leq n,~f(\Pm\cdot\gv_i)=0\\
    &\iff \forall 1\leq i\leq n,~f^{\Pm}(\gv_i)=0\\
    &\iff f^{\Pm}\in I_2(\Gm_1),
\end{align*}
which proves $(i)$. Note that $(ii)$ follows directly from this by definition of the geometric quadratic hull.
\end{proof}

As a result of this proposition, several features such as the dimension of $I_2(\Gm)$, the dimension of the geometric quadratic hull or even its smoothness do not depend on the specific generator matrix we are working with. We may sometimes omit the dependence on the generator 
matrix when we refer to invariant quantities, and write $I_2(\CC)$ instead of $I_2(\Gm)$.
\subsection{GRS codes} Many codes used in the McEliece scheme are derived from generalized Reed-Solomon codes, which we introduce below.
\begin{definition}[GRS codes]
    Let $\word{x}=(x_1,\ldots,x_n)\in\F^n$ be a vector of pairwise-distinct elements, $\word{y}\in(\F^\times)^n$, and let $r\leq n$ be an integer. The generalized Reed-Solomon (GRS) code of degree $r$, support $\word{x}$ and multiplier $\word{y}$ is defined by
    $$\GRS{r}{\word{x}}{\word{y}}=\left\{\word{y}\star f(\word{x})~|~f\in\F[X],~\deg f< r\right\},$$
    where $f(\word{x})=(f(x_1),\ldots,f(x_n))$.
\end{definition}
\begin{remark} \label{remark:vandermonde}
    A generator matrix of $\GRS{r}{\xv}{\yv}$ is given by the following truncated Vandermonde matrix:
    $$\Vm_r(\xv,\yv)=\begin{pmatrix}
        y_1 & y_2 & \ldots & y_n\\
        x_1y_1 & x_2y_2 & \ldots& x_ny_n\\
        \vdots & \vdots & \ddots &\vdots\\
        x_1^{r-1}y_1 & x_2^{r-1}y_2 & \ldots &x_n^{r-1}y_n
    \end{pmatrix}.$$
\end{remark}
A GRS code of degree $r$ has dimension $r$. It follows from the fact that univariate polynomials of degree at most $r-1$ have at most $r-1$ roots that a GRS code is always MDS, meaning that its minimum distance $d$ is equal to $n-r+1$ --- the highest possible value by Singleton's bound. Furthermore, the Welch-Berlekamp algorithm \cite{BW86} enables to decode these codes up to half their minimum distance in $O(n^3)$ operations in $\F$. 
These positive features explain why these codes have been so widely studied by researchers and engineers. However, all these nice properties mean that GRS codes are far from looking like random codes. In particular, their behavior with respect to the componentwise product is very peculiar, as we state below without proof.
\begin{proposition}\label{prop:square_grs}
$\GRS{r}{\word{x}}{\word{y}}^{\star 2}=\GRS{2r-1}{\word{x}}{\word{y}^{\star 2}}$.
\end{proposition}
Hence, the dimension of the square of a GRS code increases \textit{linearly}. In the random case \cite{CCMZ15}, the square of an $[n,k]$-code has dimension $\min\{n,\frac{k(k+1)}{2}\}$ with overwhelming probability, which means a \textit{quadratic} increase. Equation (\ref{eq:dim_I2}) shows that GRS codes having a small square means they have an unexpectedly large algebraic quadratic hull.
It is particularly visible when we look at the quadratic hull of $\Vm_r(\xv,\yv)$.
\begin{proposition}\label{prop:minors}
If $2r-1\leq n$, then $I_2(\Vm_r(\word{x},\word{y}))$ is spanned by the $2\times 2$ minors of 
\begin{equation}\label{eq:minors}
    \begin{pmatrix}
    X_0 & X_1 & \ldots & X_{r-2}\\
    X_1 & X_2 & \ldots & X_{r-1}
    \end{pmatrix}.
\end{equation}
\end{proposition}
\begin{proof}
    See for instance \cite[Propositions 1 \& 2]{LMT25}.
\end{proof}
The projective variety defined by the determinantal ideal generated by the minors of (\ref{eq:minors}) is the rational normal curve in $\PP^{r-1}$, defined as the image of the Veronese embedding 
\[\nu:\begin{cases}
    \PP^1 &\longmapsto \PP^{r-1}\\
    (u:v) &\longmapsto (v^{r-1}:uv^{r-2}:\ldots:u^{r-1}).
\end{cases}\]

\begin{remark}
    We may sometimes denote with \textit{a rational normal curve} any curve which is the image of the projective line through a map $(u:v)\mapsto(f_0(u,v):\ldots:f_{r-1}(u,v))$, where $(f_0,\ldots,f_{r-1})$ is 
    a basis of the $r$-dimensional vector space of homogeneous bivariate polynomials of degree $r-1$. Up to projective equivalence, all these curves are actually the same, which is why some authors talk about \textit{the}
    rational normal curve.
\end{remark}
The link between GRS codes and the rational normal curve provides a lot of insightful information about these codes. For example, as suggested in \cite{CMP11a}, any generator matrix of a GRS code 
can be seen as a collection of points in the projective space, through which there passes a unique rational normal curve \cite{GH78}. Computing a parametrization of the latter provides an alternative to the Sidelnikov-Shestakov attack.
\subsection{Alternant codes} The original proposal of McEliece \cite{M78} suggested implementing the scheme with binary Goppa codes, which are a subclass of the much broader family of alternant codes.
\begin{definition}[Alternant code]
Let $\word{x},\word{y}\in\Fqm^n$ be some support and multiplier, and let $r\geq 2$ be an integer. The alternant code of degree $r$, support $\word{x}$ and multiplier $\word{y}$ is defined by
$$\Alt{r}{\word{x}}{\word{y}}=\left(\GRS{r}{\word{x}}{\word{y}}^\bot\right)\cap\Fq^n.$$
\end{definition}
In this paper, we study the structure of alternant codes through their dual code. To this end, we need to introduce
the trace map associated to the Galois extension $\Fqm/\Fq$ and defined by 
$$\mathrm{Tr}:\begin{cases}\Fqm &\longrightarrow\Fq\\
    x & \longmapsto\displaystyle\sum_{j=0}^{m-1}x^{q^j}.
\end{cases}$$
Applying this map coordinatewise to any codeword is a way to construct a $q$-ary linear code given a code defined over $\Fqm$.
\begin{definition}[Trace code]
Let $\CC$ be an $[n,k]_{q^m}$-code. The \textbf{trace code} of $\CC$ is the $\Fq$-linear code defined by 
$$\Tr{\CC}\eqdef\{\Tr{\word{c}}~|~\word{c}\in\CC\},$$
where $\Tr{\word{c}}\eqdef(\Tr{c_1},\ldots,\Tr{c_n})$.
\end{definition}
Another way to construct an $\Fq$-linear code starting from an $\Fqm$-linear code $\CC$ consists of merely taking the intersection of $\CC$ with the subfield $\Fq^n$, like when we defined alternant codes. The resulting code, denoted by 
$$\CC_{|\Fq}\eqdef\CC\cap\Fq^n$$
is referred to as the \textit{subfield-subcode} of $\CC$. It turns out that those two constructions are dual to each other as stated by Delsarte's theorem.
\begin{theorem}[\cite{D75}]
    $(\CC_{|\Fq})^\perp = \Tr{\CC^\perp}$, or equivalently $(\CC^\perp)_{|\Fq}=\Tr{\CC}^\perp$.
\end{theorem}
From this, we immediately get that duals of alternant codes are trace codes of GRS codes.
\begin{corollary}
$\Alt{r}{\word{x}}{\word{y}}^\perp=\Tr{\GRS{r}{\word{x}}{\word{y}}}.$
\end{corollary}
We may thus study alternant codes as trace codes of GRS codes. In general, for an $\Fqm$-linear code $\CC$, the structure of trace code can be used to derive a generator matrix of $\Tr{\CC}$ given a generator matrix of $\CC$, which may be really helpful for establishing a link between the structure of $\CC$ and that of the trace code.
We start with a generator matrix $\Gm\in\Fqm^{r\times n}$ of $\CC$, and let $\DC=(\CC^\perp)_{|\Fq}$, so that
$$\DC=\{\word{h}\in\Fq^n~|~\forall \word{c}\in\CC,~\word{c}\cdot\word{h}=0\}.$$
Let us denote with $\word{r}_i=(g_{i,1},\ldots,g_{i,n})$ the $i$-th row of $\Gm$. By linearity of the dot product, the above condition is equivalent to $\word{h}\cdot\word{r}_i=0$ for all $0\leq i< r$. Since $\word{h}$ has coefficients in 
$\Fq$, the equation $\word{h}\cdot\word{r}_i=0$ yields $m$ equations over $\Fq$. To see this, for all coefficient $g_{i,j}$ of $\Gm$, write $g_{i,j}=g_{i,j,0}+g_{i,j,1}\alpha+\ldots+g_{i,j,m-1}\alpha^{m-1}$, with $g_{i,j,\ell}\in\Fq$, and $\alpha$ being the primitive element in $\Fqm$.
Then for any $\word{h}\in\Fq^n$,
\begin{align*}
    \word{h}\cdot\word{r}_i=0 &\iff \sum_{j=1}^n h_j g_{i,j}=0\\
    &\iff \sum_{j=1}^n\sum_{\ell=0}^{m-1}\alpha^\ell h_j g_{i,j,\ell}=0\\
    &\iff\sum_{\ell=0}^{m-1}\alpha^\ell\sum_{j=1}^n h_jg_{i,j,\ell}=0\\
    &\iff \forall~0\leq \ell< m,~\sum_{j=1}^n h_jg_{i,j,\ell}=0.
\end{align*}
As a result, the $rm\times n$ matrix obtained by taking $\Gm$ and replacing each coefficient $g_{i,j}$ by the column-vector of its $\Fq$-coordinates is a parity-check matrix of $\DC$, \textit{i.e.} a generator matrix of $\DC^\bot$. By Delsarte's theorem, 
$\DC^\perp=\Tr{\CC}$, which means that we have obtained a generator matrix of $\Tr{\CC}$. It is clear from here that $\dim_{\Fq}\Tr{\CC}\leq m\cdot\dim_{\Fqm}\CC$. Subfield-subcodes or trace codes that reach this bound will be called \textit{proper} codes. \section{The problem of recovering the structure of an alternant code} \label{section:alternants}
In the McEliece cryptosystem instantiated with the family of alternant codes, the knowledge of some support and multiplier plays the role of the private key, just as in the case of generalized Reed-Solomon codes.
A noisy codeword $\yv=\cv+\ev$ can indeed be interpreted as a noisy codeword of the underlying GRS code of dimension $n-r$. Using a support and a multiplier, one can run the Welch-Berlekamp algorithm and thus recover $\cv$,
provided that the Hamming weight of $\ev$ is strictly less than $r/2$. This is why any information about the support and the multiplier of the alternant code must again be kept secret. Although several algorithms \cite{CGGOT13,CMP11a,SS92}
can recover such a support and multiplier in polynomial time in the GRS case, this task appears to be more difficult in the generic alternant case.
\subsection{The key-recovery problem} We formulate the key-recovery problem in the following manner.
\begin{problem} \label{problem:key_recovery}
Given $\Hpub\in\Fq^{rm\times n}$ a parity-check matrix of a proper alternant code $\AC$, find a support $\word{x}$ and a multiplier $\word{y}$ such that $\AC=\Alt{r}{\word{x}}{\word{y}}$.
\end{problem}
Throughout the rest of the paper, we make use of the natural vector space identification:
\begin{equation}\label{eq:psi_alpha}
\Psi_\alpha:\begin{cases}
    \Fqm&\overset{\simeq}{\longrightarrow}\Fq^m\\
    x=\ds\sum_{j<m}x_j\alpha^j&\longmapsto (x_0,\ldots,x_{m-1}).
\end{cases}
\end{equation}
By a slight abuse of notation, given a vector $\vv=(v_0,\ldots,v_{r-1})\in\Fqm^r$, we write $\Psi_\alpha(\vv)=(\Psi_\alpha(v_0),\ldots,\Psi_\alpha(v_{r-1}))\in\Fq^{rm}$. Finally, if $\Gm\in\Fqm^{r\times n}$ is a matrix whose columns are $(\gv_1,\ldots,\gv_n)$, 
we denote by $\Psi_{\alpha}(\Gm)\in\Fq^{rm\times n}$ the matrix whose columns are $(\Psi_\alpha(\gv_1),\ldots,\Psi_\alpha(\gv_n))$.

We explained in the previous section that whenever $\Gm\in\Fqm^{r\times n}$ is a generator matrix of $\CC$, then $\Psi_\alpha(\Gm)$ is a parity-check matrix of $\CC_{|\Fq}$. In particular, $\Hsec\eqdef\Psi_\alpha(\Vm_r(\xv,\yv))$ is a valid parity-check matrix of the alternant code $\Alt{r}{\xv}{\yv}$. Moreover, the knowledge of $\Hsec$ directly yields 
the support $\xv$ and the multiplier $\yv$, and therefore the private key. If now $\Hsec=\Psi_\alpha(\Gm)$ for some other generator matrix of $\GRS{r}{\xv}{\yv}$, then one directly recovers $\Gm$ by reading the columns of $\Hsec$, and it only remains to run \cite{SS92} to recover 
an efficient decoding algorithm. We just described a whole family of easy instances of Problem \ref{problem:key_recovery}.

In reality, an attacker would not be given $\Hsec=\Psi_\alpha(\Gm)$ directly. Instead, the public key consists of $\Hpub=\Pm\cdot\Hsec$ where $\Pm$ is a random $rm\times rm$ nonsigular matrix with entries in $\Fq$. In some sense, 
the transition matrix $\Pm$ not only shuffles the basis of the underlying GRS code, but also hides the $\Fqm$-linear structure that is visible in $\Hsec$, preventing an attacker from being able to directly recover a generator matrix of the GRS code. However, Proposition \ref{prop:intrinsic} tells us 
that if we look at the quadratic hull of $\Hpub$, it has a good chance to share strong properties with that of $\Hsec$. This leads us to investigate the state of the art about the algebraic quadratic hull of an alternant code.

\subsection{The quadratic hull of an alternant code}
The dimension of the space of quadratic forms vanishing at all columns of a given parity-check matrix of an alternant code was first studied in \cite{FGOPT10}. Equivalently by Equation (\ref{eq:dim_I2}) one can study the dimension of the square of the dual code, as it was first noticed in \cite{MP12} 
and fully investigated in \cite{MT21}.

\begin{theorem}[\cite{MT21}, Theorem 19] \label{thm:ea}
Let $\CC=\Alt{r}{\word{x}}{\word{y}}^\perp$ be a proper dual alternant code. Then
\begin{equation}\label{eq:ea}
    \dim I_2(\CC)\geq \dfrac{m}{2}(r-1)\left((2\ea+1)r-2\dfrac{q^{\ea+1}-1}{q-1}\right),
\end{equation}
where $\ea=\lfloor \log_q(r-1)\rfloor$.
\end{theorem}
Note that we did not specify any generator matrix in the above statement, as the result is independent of such a choice. Even though the above result is an inequality, it turns out that one can confidently predict when equality holds.
\begin{heur}[\cite{FGOPT10}]\label{heur:fgopt}
Assume that $\xv$ and $\yv$ are chosen independently at random. Then equality is reached in Inequality (\ref{eq:ea}) as soon as the right-hand side exceeds $\binom{rm+1}{2}-n$.
\end{heur}
This is the high-rate regime corresponding to the so called \textit{square-distinguishability} of (duals of) alternant codes. In general, we say that alternant codes are square-distinguishable when the square of their dual code 
has dimension lower than $\min\{n,\binom{rm+1}{2}\}$, \textit{i.e.} when the dimension of the square of the dual code is unexpectedly low compared to what we would obtain from a random code having the same parameters.

When we further assume that $r\leq q$, we see that $\ea=0$ and the dimension of $I_2(\CC)$ simplifies as 
$$\dim I_2(\CC)=m\binom{r-1}{2},$$
which means that in such a case, the dimension of the vanishing ideal at degree 2 of an alternant code is exactly $m$ times that of the underlying GRS code, which suggests a link between the algebraic quadratic hull of the dual of an alternant code and that of the 
underlying GRS code. We use the geometric framework to explain this. Let $\cY$ be the geometric quadratic hull of $\GRS{r}{\word{x}}{\word{y}}$ with respect to $\Vm_r(\xv,\yv)$. Recall that $\cY$, when seen as a projective variety, is the rational normal curve and its defining ideal 
is generated by the $2\times 2$ minors of (\ref{eq:minors}).
A point $P=(p_0:\ldots:p_{r-1})\in\PP^{r-1}(\Fqm)$ belongs to $\cY$ if and only if the vector $(p_0,\ldots,p_{r-1})$ satisfies these homogeneous quadratic equations. As we will explain in the following section, this condition boils down to the whole vector of
coordinates $\Psi_\alpha(p_0,\ldots,p_{r-1})\in\Fq^{rm}$ satisfying $m$ quadratic relations over $\Fq$. In other words, any quadratic equation satisfied by the columns of $\Vm_r(\xv,\yv)$ yields $m$ quadratic equations satisfied by the columns of the parity-check matrix $\Hsec=\Psi_\alpha(\Vm_r(\xv,\yv))$ of $\Alt{r}{\xv}{\yv}$.
This phenomenon highlights a link between the quadratic hull of an $\Fqm$-linear code and that of its trace code. This is the concept of \textit{affine Weil restriction}, to which the following section is dedicated. \section{Affine Weil restriction} \label{section:weil}
The idea of Weil restriction is to construct a variety defined over $\Fq$ given a variety defined over $\Fqm$ by splitting the variables of the defining equations according to their $\Fq$-coordinates.
An algebraic variety $V\subset\Fqm^r$ would thus give rise to a variety $W\eqdef\Psi_{\alpha}(V)\subset\Fq^{rm}$. Weil restriction
is the process through which we obtain the defining ideal of $W$ given the defining ideal of $V$.
\subsection{Definition and first properties} Let $\Rm=\Fqm[X_0,\ldots,X_{r-1}]$ be the polynomial ring over $\Fqm$ in $r\geq 2$ variables. Besides, consider the two polynomial rings
$$\Sm=\Fq[x_{i,j}~|~0\leq i<r,~0\leq j< m],$$
and
$$\Sm\otimes_{\Fq}\Fqm=\Fqm[x_{i,j}~|~0\leq i<r,~0\leq j< m].$$
This last ring, of which $\Sm$ is a subring, enables us to define the following as a homomorphism of graded $\Fqm$-algebras:
\begin{equation}\label{eq:hom}
\Phi:\begin{cases}\Rm &\longrightarrow\Sm\otimes_{\Fq}\Fqm\\
    X_i&\longmapsto \sum_{j=1}^r\alpha^jx_{i,j}.
\end{cases}
\end{equation}
Applying $\Phi$ corresponds to splitting each variable $X_i$ according to its $\Fq$-coordinates. The resulting polynomial has coefficients in $\Fqm$, so we may again gather its $\Fq$-coordinates. In other words,
for any $f\in\Rm$, there are unique $\Phi_1(f),\ldots,\Phi_m(f)\in\Sm$ such that
\[\Phi(f)=\displaystyle\sum_{j=0}^{m-1}\alpha^j\Phi_j(f).\]
In particular, each map $\Phi_j:\Rm\rightarrow\Sm$ is an $\Fq$-linear map.
\begin{definition}[Weil restriction]\label{def:weil}
The Weil restriction of an ideal $I\subset\Rm$ is defined by
$$\Res{\Fqm/\Fq}{I}\eqdef\langle \Phi_j(f)~|~f\in I,~0\leq j< m\rangle\subset\Sm.$$
\end{definition}
The main property at the core of Weil restriction is the correspondence between rational points.
\begin{proposition} \label{prop:rational}
    Let $I\subset\Rm$ be an ideal, and let $J=\Resq{I}$. Then 
    $$V_{\Fq}(J)=\Psi_\alpha(V_{\Fqm}(I)).$$
\end{proposition}
\begin{proof}
Let $P\in\Fqm^r$. The following equivalences hold:
\begin{align*}
\forall f\in I,~f(P)=0 & \iff \forall f\in I,~\Phi(f)(\Psi_{\alpha}(P))=0\\
&\iff \forall f\in I,~\forall 0\leq j< m,~\Phi_j(f)(\Psi_{\alpha}(P))=0 \text{~By identification} \\
&\iff \forall g \in J,~g(\Psi_{\alpha}(P))=0,
\end{align*}
Hence $P\in V_{\Fqm}(I)\iff \Psi_{\alpha}(P)\in V_{\Fq}(J)$, and thus $\Psi_{\alpha}(V_{\Fqm}(I))\subseteq V_{\Fq}(J)$. By taking $Q\in\Fq^{rm}$ and reading the above equivalences backwards with $\Psi_\alpha^{-1}(Q)$ playing the role of $P$, we obtain the converse inclusion.
\end{proof}

Weil restriction can be defined as a functor from the category of varieties defined over $\Fqm$ to that of varieties defined over $\Fq$. This means that we can define the Weil restriction of a morphism $f:V_1\rightarrow V_2$ as well, where $V_1\subset\Fqm^{r}$ and $V_2\subset\Fqm^{s}$ are algebraic varieties.
We will restrict ourselves to the case where $f$ is defined by a collection of polynomials $(f_1,\ldots,f_s)$, with each $f_i\in\Fqm[V_1]$. In such a case, one can take any representative of each $f_i$ in $\Rm$ and apply the map $\Phi$, hence getting polynomials $\Phi_0(f_i),\ldots,\Phi_{m-1}(f_i)\in\Sm$.
\begin{proposition}
    The collection of polynomials $(\Phi_0(f_1),\ldots,\Phi_{m-1}(f_1),\ldots,\Phi_0(f_s),\ldots,\Phi_{m-1}(f_s))$ is a well-defined rational function $W_1\rightarrow W_2$, where $W_i\eqdef\Psi_\alpha(V_i),~i\in\{1,2\}$. We refer to it as the Weil restriction of $f$.
\end{proposition}
\begin{proof}
Let $I\subset\Rm$ be the ideal defining $V_1$. By definition of Weil restriction, if $g\in I$, then all the $\Phi_j(g)$'s belong to $J\eqdef\Resq{I}$. As a consequence, if $f\in\Fqm[V_1]$ is defined
up to an element of $I$, then all the $\Phi_j(f)$'s are well-defined up to an element of $J$. This proves that the collection of polynomials given in the proposition define a rational function $W_1\rightarrow W_2$.
\end{proof}
\begin{remark}
This is again the good notion for Weil restriction, as we can easily see that the following
\[
\begin{tikzcd}[column sep=2cm]
V_1 \arrow[r, "f"] \arrow[d, "\Psi_\alpha"'] & V_2 \arrow[d, "\Psi_\alpha"] \\
W_1 \arrow[r, "\Resq{f}"'] & W_2
\end{tikzcd}
\]
is a commutative diagram.
\end{remark}
We now investigate the case of $\Fqm$-linear maps, as it will be central in the following.
First, let us introduce the matrix $\Jm\in\Fq^{m\times m}$ defined as the matrix of the $\Fq$-linear map
$$\mu_{\alpha}:\begin{cases}
    \Fqm&\longrightarrow\Fqm\\
    x &\longmapsto \alpha x,
\end{cases}$$
with respect to the monomial basis $(1,\alpha,\ldots,\alpha^{m-1})$ of $\Fqm$. Note that $\Jm$ is nothing but the companion matrix of the minimal polynomial $\Pi_\alpha$ of $\alpha$ over $\Fq$. Given any field element $x=x_0+x_1\alpha+\ldots+x_{m-1}\alpha^{m-1}$, the matrix $x_0\Im_m+x_1\Jm+\ldots+x_{m-1}\Jm^{m-1}$ is the matrix of the multiplication by $x$ in the monomial
basis $(1,\alpha,\ldots,\alpha^{m-1})$. We denote it by $\Mata(x)$. All in all, the map
$$\Mat_{\alpha}:\begin{cases}
    \Fqm &\longrightarrow \Fq[\Jm]\\
    x &\longmapsto \Mata(x)
\end{cases}$$
defines an isomorphism of $\Fq$-algebras. These matrices are the natural way of seeing multiplications in terms of $\Fq$-coordinates, as
\begin{equation}\label{eq:mata}
\forall x,y\in\Fqm,~\Psi_\alpha(xy)=\Mata(x)\cdot\Psi_\alpha(y).
\end{equation}
Furthermore, one can identify the image of $\Mat_{\alpha}$ using the following criterion.
\begin{proposition}\label{prop:cond}
    Let $\Bm\in\Fq^{m\times m}$. Then $\Bm\in\Fq[\Jm]$ if and only if $\Bm\Jm=\Jm\Bm$.
\end{proposition}
\begin{proof}
    It is clear that a polynomial in $\Jm$ commutes with $\Jm$. Conversely, let $\Bm\in\Fq^{m\times m}$ and assume that $\Bm$ and $\Jm$ commute.
    The minimal polynomial $\Pi_{\Jm} $ of $\Jm $ over $\Fq$ is that of $\alpha$, and is therefore irreducible of degree $m$. As $\Fqm/\Fq$ is a Galois extension, $\Pi_{\Jm} $ splits into linear factors over $\Fqm$. More precisely,
    \[\Pi_{\Jm} =\ds\prod_{j=0}^{m-1}(X-\alpha^{q^j}).\]
    This implies that $\Jm $ is diagonalizable over $\Fqm$, and there exists $\Pm\in\GL_m(\Fqm)$ such that
    $$\Pm\Jm\Pm^{-1}=\Delta_\alpha\eqdef\mathrm{Diag}(\alpha,\alpha^q,\ldots,\alpha^{q^{m-1}}),$$
    the diagonal entries being pairwise distinct. Since $\Bm$ commutes with $\Jm$, we see that $\Pm\Bm\Pm^{-1}$ commutes with $\Delta_\alpha$ and consequently stabilizes its eigenspaces. As the latter are lines, we conclude
    that $\Pm\Bm\Pm^{-1}$ is also a diagonal matrix and there exists $(\beta_1,\ldots,\beta_m)\in\Fqm^m$ such that $\Pm\Bm\Pm^{-1}=\mathrm{Diag}(\beta_1,\ldots,\beta_m)$. Now let $f\in\Fqm[X]$ be an interpolating polynomial such that $f(\alpha^{q^j})=\beta_j$
    for all $j$. We see that
    $$f(\Pm\Jm\Pm^{-1})=\Pm f(\Jm )\Pm^{-1}=\mathrm{Diag}(f(\alpha),\ldots,f(\alpha^{q^{m-1}}))=\mathrm{Diag}(\beta_1,\ldots,\beta_m)=\Pm\Bm\Pm^{-1}.$$
    Thus, $\Bm=f(\Jm )$. By writing $f=f_0+\alpha f_1+\ldots+\alpha^{m-1}f_{m-1}$, with $f_j\in\Fq[X]$, and then proceeding through identification (which is possible as both $\Bm$ and $\Jm $ lie in $\Fq^{m\times m}$) we see that $\Bm=f_0(\Jm )$ which ends the proof.
\end{proof}
Equation (\ref{eq:mata}) essentially states that $\Mata(x)$ is the matrix of the Weil restriction of the multiplication map associated to $x$. This can be naturally generalized to higher dimensions. More specifically, we define below the Weil restriction of a matrix. 
\begin{definition}
    Let
    $$\Bm=(a_{i,j})_{\substack{1 \leq i \leq s \\ 1 \leq j \leq r}}\in\Fqm^{s\times r}$$
    be a matrix. We denote by $\Resq{\Bm}$ the matrix of the Weil restriction of the map $\vv\mapsto\Bm\vv$ in the canonical basis, \textit{i.e.}
    $$\forall \vv\in\Fqm^r,~\Psi_\alpha(\Bm\cdot\vv)=\Resq{\Bm}\cdot\Psi_\alpha(\vv).$$
\end{definition}
\begin{remark}
By applying Equation (\ref{eq:mata}) coefficient-wise, we see that
$$\Resq{\Bm}=\begin{pmatrix}
    \Mata(a_{1,1}) & \ldots & \Mata(a_{1,r})\\
    \vdots & \ddots & \vdots \\
    \Mata(a_{s,1}) & \ldots & \Mata(a_{s,r})
\end{pmatrix}\in\Fq^{sm\times rm}.$$
\end{remark}
For any non-negative integer $k$, let $\Jm_k=\mathrm{Diag}(\Jm,\ldots,\Jm)$. Note that $\Jm_k$ is the Weil restriction of the homothety of $\Fqm^k$ defined by $\alpha$, \textit{i.e.} the scalar
multiplication by $\alpha$. Matrices of the form $\Resq{\Bm}$ can be identified in $\Fq^{sm\times rm}$ using the following algebraic criterion, which may be thought of as a generalization of Proposition \ref{prop:cond}.
\begin{proposition} \label{prop:cond_r}
    Let $\Am\in\Fq^{sm\times rm}$. The following are equivalent:
    \begin{itemize}
        \item[$(i)$] $\Am\Jm_r=\Jm_s\Am$;
        \item[$(ii)$] there exists $\Bm\in\Fqm^{s\times r}$ such that $\Am=\Resq{\Bm}$.
    \end{itemize}
\end{proposition}
\begin{proof}
The existence of $\Bm\in\Fqm^{s\times r}$ such that $\Am=\Resq{\Bm}$ is equivalent to each block $\Am_{i,j}\in\Fq^{m\times m}$ being in the image of $\Mata$, which by Proposition \ref{prop:cond} boils down to the $\Am_{i,j}$'s commuting with $\Jm$. Gathering these commutativity conditions is exactly equivalent to $(i)$.
\end{proof}

\subsection{Weil-properness} We now have introduced the necessary framework for stating as a proposition the link between the quadratic hull of an alternant code and that of the underlying GRS code.
\begin{proposition}\label{prop:inclusion}
Let $\CC\subset\Fqm^n$ be a proper code. Let $\Gm\in\Fqm^{r\times n}$ be a generator matrix of $\CC$ and let $\Hsec=\Psi_{\alpha}(\Gm)$, which is the secret parity-check matrix of $\Tr{\CC}$. Then
\begin{equation}\label{eq:inclusion}
\Resq{\langle I_2(\Gm)\rangle}\subseteq \langle I_2(\Hsec)\rangle,
\end{equation}
and
\begin{equation}\label{eq:inclusion2}
    V_2(\Hsec)\subseteq \Psi_\alpha(V_2(\Gm)).
\end{equation}
\end{proposition}
\begin{proof}
By Proposition \ref{prop:rational}, for any polynomial $f\in I_2(\Gm)$, the $\Phi_j(f)$'s all vanish at the columns of $\Hsec$. As a consequence,
\[\mathrm{Span}_{\Fq}\{\Phi_j(f)~|~f\in I_2(\Gm),~0\leq j<m\}\subseteq I_2(\Hsec),\]
which is a refinement of (\ref{eq:inclusion}). Taking the varieties reverses the inclusion, which leads to (\ref{eq:inclusion2}) by Proposition \ref{prop:rational}.
\end{proof}
The above proposition only states inclusions, while equalities will be needed in the following.
\begin{definition}[Weil-properness]
    A linear code $\DC=\Tr{\CC}$ is said to be \textbf{Weil-proper} if and only if Inclusion (\ref{eq:inclusion}) is an equality.
\end{definition}
Determining whether a linear code is Weil-proper may sometimes be done by just measuring the dimension of its vanishing ideal at degree $2$. For alternant codes, we have the following result.

\begin{proposition}\label{prop:weil_proper_alt}
Under Heuristic \ref{heur:fgopt}, when $r\leq q$, a generic $q$-ary alternant code of degree $r$ achieving equality in (\ref{eq:ea}) is Weil-proper.
\end{proposition}
For the proof of Proposition \ref{prop:weil_proper_alt} we will need the following lemma. The proof of the lemma is quite technical and not really relevant here, which is why we decided to move it to the appendix.

\begin{restatable}{lemma}{genset}
\label{lemma:genset}
Let $I\subset\Rm$ be a homogeneous ideal. If $(f_1,\ldots,f_N)$ is a minimal set of generators for $I$, then the sequence $(\Phi_j(f_i))_{i,j}$ is a minimal set of generators of $\Res{\Fqm/\Fq}{I}$.
\end{restatable}
In case of Weil-properness, it is indeed clear that if $(f_1,\ldots,f_N)$ is a basis of $I_2(\Gm)$, then $(\Phi_j(f_i))_{i,j}$ generates $I_2(\Psi_\alpha(\Gm))$. What is much less clear however is that there are no linear dependencies between the $\Phi_j(f_i)$'s over $\Fq$.
\begin{corollary}\label{cor:weil_proper_condition}
    Let $\DC=\Tr{\CC}$. Then $\DC$ is Weil-proper if, and only if
    $$\dim_{\Fq}I_2(\DC)=m\dim_{\Fqm}I_2(\CC).$$
\end{corollary}
\begin{proof}
By Proposition \ref{prop:inclusion} and Lemma \ref{lemma:genset}, the Weil restriction of $I_2(\CC)$ is an $m\dim_{\Fqm}I_2(\CC)$-dimensional subspace of $I_2(\DC)$. Equality therefore holds if and only if $I_2(\DC)$ has dimension $m\dim_{\Fqm}I_2(\CC)$.
\end{proof}
Gathering all these results enables us to prove Proposition \ref{prop:weil_proper_alt} under Heuristic \ref{heur:fgopt}.
\begin{proof}[Proof of Proposition \ref{prop:weil_proper_alt}]
When $r\leq q$, the vanishing ideal at degree $2$ of a generic alternant code in the square-distinguishable regime equals $m$ times the dimension of the vanishing ideal at degree $2$ of the underlying GRS code. By Corollary \ref{cor:weil_proper_condition} we get that such an alternant code
is Weil-proper.
\end{proof}
Lemma \ref{lemma:genset} also allows us to derive a necessary condition for a code to be Weil proper in general in terms of regime of parameters.

\begin{corollary} \label{cor:regime}
Let $\DC=\Tr{\CC}$ be a proper linear code, \textit{i.e.} $\dim_{\Fq}\DC=rm$ where $r=\dim_{\Fqm}\CC$. If $\DC$ is Weil proper, then
$$n\geq \binom{rm+1}{2}-m\dim_{\Fqm} I_2(\CC).$$
\end{corollary}

\begin{proof}
By Corollary \ref{cor:weil_proper_condition}, we have
$$\binom{rm+1}{2}-m\dim_{\Fqm}I_2(\CC)=\binom{rm+1}{2}-\dim_{\Fq}I_2(\DC)=\dim\DC^{\star 2}\leq n.$$
\end{proof}

In the regime where an alternant code is Weil-proper, the Weil restriction structure of the quadratic hull may be visible even if we only have access to $\Hpub=\Pm\cdot\Hsec$, thanks to Proposition \ref{prop:intrinsic}. Our goal is to determine which properties of $I_2(\Hsec)$, related to its Weil restriction structure, are preserved by linear transformation and therefore
still detectable in $I_2(\Hpub)$. This naturally leads to the problem of distinguishing affine varieties that are the Weil restriction of a smaller variety over a larger field.

\subsection{Distinguishing Weil restrictions} Our strategy for finding a criterion that distinguishes Weil restrictions from other varieties consists in first looking at the case of vector subspaces, and then generalizing to algebraic varieties using tangent spaces.

Let $V\subset\Fqm^r$ be some vector subspace. Intuitively, $\Psi_\alpha(V)$ not only lists the points of $V$ in terms of $\Fq$-coordinates, but also somehow reflects the $\Fqm$-linearity of $V$. More formally, the fact that if $\vv\in V$, then $\alpha\cdot \vv\in V$ must be visible in $W$. Indeed, as $\Jm_r$ is the Weil restriction of the scalar multiplication by $\alpha$, we see that
\begin{equation} \label{eq:linearity}
    \alpha\cdot\vv\in V\iff \Jm_r\cdot\Psi_\alpha(\vv)\in W.
\end{equation}
The following proposition states that this can be used to identify Weil restriction of vector spaces.

\begin{lemma} \label{lemma:equivalence}
    Let $W\subset \Fq^{rm}$ be a vector subspace. The following are equivalent:
    \begin{itemize}
    \item[$(i)$] $W$ is $\Jm _r$-invariant, \textit{i.e.} $\forall \word{x}\in W,~\Jm_r\word{x}\in W$;
    \item[$(ii)$] There exists some subspace $V\subset\Fqm^r$ such that $W=\Psi_\alpha(V)$.
    \end{itemize}
    \end{lemma}
    \begin{proof}
        Let $W\subset\Fq^{rm}$ be a vector subspace, and set $V=\Psi_\alpha^{-1}(W)$, which is an $\Fq$-vector subspace of $\Fqm^r$ a priori. It suffices to prove that $W$ is $\Jm_r$-invariant if and only if $V$ is an
        $\Fqm$-vector subspace of $\Fqm^r$. Since $V$ is already an $\Fq$-vector space, it is an $\Fqm$-vector space if, and only if
        $$\forall \vv\in V,~\alpha\vv\in V,$$
        which by Equivalence (\ref{eq:linearity}) is equivalent to
        $$\forall \vv\in V,~\Jm_r\Psi_\alpha(\vv)\in W.$$
        Finally, since $\Psi_\alpha$ is a bijection between $V$ and $W$, we see that $V$ is an $\Fqm$-vector subspace if and only if
        $$\forall \wv\in W,~\Jm_r\wv\in W,$$
        which proves the proposition.
    \end{proof}
\begin{remark}
    Lemma \ref{lemma:equivalence} is actually a well-known result from the theory of matrix rank-metric codes. More precisely, given a $k$-dimensional $\Fqm$-linear code 
    $\CC\subset\Fqm^n$ endowed with the rank-metric, one can build a code $\Cmat$ from a basis $\{\vv_1,\ldots,\vv_k\}$ of $\CC$ by defining 
    $$\Cmat=\mathrm{Span}_{\Fq}\{\Mm_{1,0},\ldots,\Mm_{1,m-1},\ldots,\Mm_{k,0},\ldots,\Mm_{k,m-1}\},$$
    where $\Mm_{i,j}=(\Psi_\alpha(\alpha^jv_{i,1})|\ldots|\Psi_\alpha(\alpha^jv_{i,n}))\in\Fq^{m\times n}$, with $\vv_i=(v_{i,1},\ldots,v_{i,n})$. As $\CC$ is an $\Fqm$-linear code, 
    we have 
    $$\forall \Mm\in\Cmat,~\Jm\Mm\in\Cmat.$$
    More generally, an $\Fq$-vector subspace $\mathscr{M}\subset\Fq^{m\times n}$ is $\Fqm$-linear, \textit{i.e.} built using the above process, if and only if $\Jm\Mm\in\mathscr{M}$ for all $\Mm\in\mathscr{M}$.
\end{remark}
\begin{definition}[Stabilizer]
For any vector space $W\subset\Fq^{rm}$, we denote by $\St{W}$ the set of matrices $\Am\in\Fq^{rm\times rm}$ such that $\Am\cdot W\subseteq W$. Note that $\St{W}$ is a subalgebra of the $\Fq$-algebra $\Fq^{rm\times rm}$.
\end{definition}
\begin{remark}
Lemma \ref{lemma:equivalence} amounts to say that a subspace $W$ is a Weil restriction if and only if we have $\Fq[\Jm_r]\subset\St{W}$.
\end{remark}
Distinguishing Weil restrictions is therefore a solved problem when it comes to vector subspaces. In order to generalize our approach for algebraic varieties, we need some linear data associated to varieties. This is exactly the role played by tangent spaces.
\begin{definition}
Let $P\in V\eqdef V_{\Fqm}(I)$ where $I\subset\Rm$ is an ideal. The tangent space of $V$ at $P$ is defined by
$$T_P V=\left\{\word{h}\in\Fqm^r~|~\forall f\in I,~\ds\sum_{i=0}^{r-1} h_i\partial_i f(P)=0\right\}.$$
As $I$ is finitely generated, $T_P V$ may be computed as the kernel of the Jacobian matrix of a list of generators.
\end{definition}
A vector subspace of $\Fq^{rm}$ is a Weil restriction if and only if it is globally invariant under the action of $\Jm_r$. When a variety $W$ is the Weil restriction of some other variety $V$, then one can expect that its tangent spaces are also the Weil restriction of some vector subspace of $\Fqm^r$. This would enable us
to use our criterion to determine whether a variety is a Weil restriction. It turns out that it is true, thanks to the commutativity between Weil restrictions and tangent spaces.
\begin{proposition}\label{prop:tangent}
    Let $V=V_{\Fqm}(I)\subset\Fqm^r$ be an algebraic variety of defining ideal $I\subset\Rm$. Let $W=\Psi_\alpha(V)$, $P\in V$ and $Q=\Psi_\alpha(P)$. Then
    $$T_Q W=\Psi_\alpha(T_P V).$$
\end{proposition}
\begin{proof}
First, notice that by the rules of derivation, we have
\begin{equation} \label{eq:derivation}
\partial_{ij}\Phi(f)=\alpha^j\Phi(\partial_i f),
\end{equation}
for all $f\in\Rm$ and indices $0\leq i<r,0\leq j<m$. Now, let $\word{h}=(h_{i,j})\in\Fq^{rm}$. We have
\begin{align*}
\word{h}\in\Psi_\alpha(T_P V)&\iff \forall f\in I,~\displaystyle\sum_{i=0}^{r-1}\left(\sum_{j=0}^{m-1}h_{i,j}\alpha^j\right)\partial_i f(P)=0\\
&\iff \forall f\in I,~\sum_{i=0}^{r-1}\sum_{j=0}^{m-1} h_{i,j}\alpha^j\Phi(\partial_i f)(Q)=0 \text{ Since }\partial_i f(P)=\Phi(\partial_i f)(Q) \\
&\iff \forall f\in I,~\sum_{i=0}^{r-1}\sum_{j=0}^{m-1} h_{i,j}\partial_{ij}\Phi(f)(Q)=0 \text{ By Equation (\ref{eq:derivation})} \\
&\iff \forall f\in I,~\forall 0\leq \ell<m,~\sum_{i=0}^{r-1} h_{i,j}\partial_{ij}\Phi_\ell(f)(Q)=0\\
&\iff \word{h}\in T_Q W.
\end{align*}
\end{proof}
This gives us a necessary condition for an algebraic variety to be a Weil restriction.
\begin{corollary}
Let $W\subset\Fq^{rm}$ be an algebraic variety. If there exists an algebraic variety $V\subset\Fqm^r$ such that $W=\Psi_\alpha(V)$, then for all $Q\in W$, the tangent space
$T_Q W\subset\Fq^{rm}$ is $\Jm_r$-invariant.
\end{corollary}
Back to linear codes, we obtain the following corollary.
\begin{corollary} \label{corollary:invariance_hsec}
    Let $\Hsec=\Psi_\alpha(\Vm_r(\xv,\yv))$ be the secret parity-check matrix of a Weil-proper alternant code $\Alt{r}{\word{x}}{\word{y}}$. For all $P\in \Vsec\eqdef V_2(\Hsec)$, the tangent space
    $T_P\Vsec$ is $\Jm_r$-invariant.
\end{corollary}

\subsection{Weil-preserving transformations} Corollary \ref{corollary:invariance_hsec} establishes a distinguishing property for the quadratic hull of a Weil-proper alternant code $\Alt{r}{\xv}{\yv}$ with respect to a secret parity-check matrix $\Hsec$, \textit{i.e.}
of the form $\Psi_\alpha(\Gm)$ where $\Gm$ is a generator matrix of the underlying GRS code. Again, denote by $\Hpub$ the corresponding public key, which is related to the private key by the relation $\Hpub=\Pm\Hsec$ for some secret nonsingular
$rm\times rm$ matrix $\Pm$, which is nothing but a change of basis. We get the following proposition given how a change of basis acts on linear maps.
\begin{proposition} \label{prop:invariance}
    If $\Alt{r}{\xv}{\yv}$ is Weil-proper, then for all $Q\in \Vpub\eqdef V_2(\Hpub)$, the tangent space $T_Q \Vpub$ is globally invariant under the action of $\Pm\Jm_r\Pm^{-1}$.
\end{proposition}
\begin{proof}
    We know from Proposition \ref{prop:intrinsic} that $\Vpub=\Pm\cdot\Vsec$, which implies that $T_Q\Vpub=\Pm\cdot T_{\Pm^{-1}Q}\Vsec$. Since $T_{\Pm^{-1}Q}\Vsec$ is $\Jm_r$-invariant, $T_Q\Vpub$ is $\Pm\Jm_r\Pm^{-1}$-invariant.
\end{proof}
All tangent spaces of the public variety $\Vpub$ share the property of being invariant under the very same linear operator. This will be the first cornerstone of our attack, which we will detail in the following section. The second core idea behind the attack consists in
finding another transition matrix than $\Pm$ that directly leaks a generator matrix of the underlying GRS code --- or, as we will see, one of its conjugates. Such a transition matrix would also map the variety $\Vsec\eqdef V_2(\Hsec)$ onto another algebraic variety, itself being linked to the quadratic hull of a GRS code through
Weil restriction. We then see that the set of such transition matrices are exactly those that map Weil restrictions onto other Weil restrictions. The following theorem is an exhaustive description of such linear transformations.

Weil restrictions of invertible $r\times r$ matrices over $\Fqm$ are natural candidates for such transition matrices. If $\Am=\Resq{\Bm}$, then $\Am$ indeed maps a Weil restriction $W=\Psi_\alpha(V)$ onto that of $\Bm\cdot V$. Another type of transformations mapping Weil restrictions onto Weil restrictions is
given by the Frobenius automorphism. Let $\Thetam$ be the matrix of the Frobenius automorphism $\theta:x\mapsto x^q$ in the monomial basis $(1,\alpha,\ldots,\alpha^{m-1})$, and define $\Thetam_r=\mathrm{Diag}(\Thetam,\ldots,\Thetam)\in\Fq^{rm\times rm}$. Then $\Thetam_r$ maps $W=\Psi_\alpha(V)$ onto the Weil restriction of $V^q=\{(v_1^q,\ldots,v_r^q)~|~\vv\in V\}$, which is also an algebraic variety.
The following theorem essentially states that these two examples generate all possible matrices mapping Weil restrictions onto Weil restrictions.
\begin{theorem} \label{thm:group} Let $\Group$ be the set of invertible $rm\times rm$ matrices over $\Fq$ that map Weil restrictions of subvarieties of $\Fqm^r$ onto Weil restrictions of subvarieties of $\Fqm^r$. Then
$$\Group=\{\Resq{\Bm}\cdot\Thetam_r^j~|~\Bm\in\GL_r(\Fqm),~0\leq j<m\}.$$
\end{theorem}

The goal of the attack will be to find a matrix $\Qm$ such that $\Qm\Pm\in\Group$, so that the columns of $\Qm\Hpub=\Qm\Pm\Hsec$ belong to the Weil restriction of some rational normal curve.
Before proving Theorem \ref{thm:group}, we need to give two auxiliary results. The first gives a better description of the action of the Frobenius automorphism in terms of $\Fq$-coordinates. The matrix $\Thetam$
represents the Frobenius map in the following way:
$$\forall x\in\Fqm,~\Psi_\alpha(x^q)=\Thetam\cdot\Psi_\alpha(x).$$
Meanwhile, the element $x\in\Fqm$ is also identified with the matrix $\Mata(x)$. Since $\Mata$ is a field isomorphism, we have
$$\forall x\in\Fqm,~\Mata(x^q)=\Mata(x)^q.$$
There is a way of expressing the above property using the matrix $\Thetam$ as stated in our first lemma.

\begin{lemma}\label{lemma:conjugation}
    For all $x\in\Fqm,~\Mata(x)^q=\Thetam\cdot\Mata(x)\cdot\Thetam^{-1}$.
\end{lemma}
\begin{proof}
    Let $x\in\Fqm$. For all $y\in\Fqm$, we have
    \begin{align*}
    \Thetam\cdot\Mata(x)\cdot\Psi_\alpha(y) &=\Thetam\cdot\Psi_\alpha(xy)\\
    &=\Psi_\alpha((xy)^q)\\
    &=\Psi_\alpha(x^q y^q)\\
    &=\Mata(x^q)\cdot\Psi_\alpha(y^q)\\
    &=\Mata(x)^q\cdot \Thetam\cdot\Psi_\alpha(y),
    \end{align*}
    and this holds for any choice of $y\in\Fqm$. As a result, $\Thetam\cdot\Mata(x)=\Mata(x)^q\cdot\Thetam$.
\end{proof}
\begin{remark}
    Lemma \ref{lemma:conjugation} can be summed up in one sentence, essentially saying that conjugation by the matrix $\Thetam$ and Galois conjugation by the Frobenius automorphism
    boil down to the very same operation.
\end{remark}
Lemma \ref{lemma:equivalence} states that vector subspaces that are the Weil restriction of another vector space are stabilized by $\Jm_r$, and therefore by any polynomial in $\Jm_r$. Such a matrix
is nothing but the Weil restriction of an $\Fqm$-homothety. The second ingredient in the proof of Theorem \ref{thm:group} is the converse: a matrix that stabilizes all vector subspaces that are
a Weil restriction is a polynomial in $\Jm_r$.
\begin{lemma}\label{lemma:stabilizers}
    Let $\Am\in\Fq^{rm\times rm}$. If $\Am$ stabilizes all vector subspaces of $\Fq^{rm}$ that are the Weil restriction of a subspace of $\Fqm^r$, then $\Am\in\Fq[\Jm_r]$.
\end{lemma}
\begin{proof}
    Let us introduce the $\Fq$-linear endomorphism $a$ of $\Fqm^r$ represented by the matrix $\Am$, \textit{i.e.} $a$ is defined by
    $$\forall \vv\in\Fqm^r,~a(\vv)=\Psi_\alpha^{-1}(\Am\cdot\Psi_\alpha(\vv)).$$
    Since $\Am$ stabilizes vector subspaces that are Weil restrictions, it stabilizes in particular Weil restrictions of $\Fqm$-lines. Equivalently, the $\Fq$-linear endomorphism $a$ maps any vector $\vv\in\Fqm^r$ onto some $\lambda_{\vv} \vv$, with $\lambda_{\vv}\in\Fqm$.
    If now we take two vectors $\uv,\vv\in\Fqm^r$, then by $\Fq$-linearity we have
    $$a(\uv+\vv)=\lambda_{\uv+\vv}(\uv+\vv)=a(\uv)+a(\vv)=\lambda_{\uv}\uv+\lambda_{\vv}\vv.$$
    As $r\geq 2,$ we can take $\uv$ and $\vv$ linearly independent over $\Fqm$, which then implies
    $$\lambda_{\uv+\vv}=\lambda_{\uv}=\lambda_{\vv}.$$
    Finally, by noticing that the same reasoning holds if we replace $\vv$ with $\beta\vv$ for some nonzero $\beta\in\Fqm$, we conclude that $\lambda_{\beta\vv}=\lambda_{\vv}$. All in all, there exists a \textit{unique} $\lambda\in\Fqm$ such that
    $$\forall \vv\in\Fqm^r,~a(\vv)=\lambda\vv,$$
    which indeed means that $a$ is an $\Fqm$-linear homothety. In terms of matrices, this implies that $\Am=f(\Jm_r)$, the polynomial $f\in\Fq[X]$ being the one that determines $\lambda$.
\end{proof}
\begin{remark}
    The assumption $r\geq 2$ that we made at the very beginning of this section is mandatory here, not only for the proof to work but also for the result to hold. If indeed $r=1$, then the only $\Fqm$-subspaces of $\Fqm^r$ are $\{0\}$ and 
    $\Fqm$, whose Weil restriction are respectively $\{0\}$ and $\Fq^m$. All matrices $\Am\in\Fq^{m\times m}$ stabilize these spaces, and therefore the implication stated by Lemma \ref{lemma:stabilizers} is untrue in this case.
\end{remark}
We are now ready to prove Theorem \ref{thm:group}.
\begin{proof}[Proof of Theorem \ref{thm:group}]
Clearly any matrix $\Am$ of the form $\Resq{\Bm}\cdot\Thetam^j$ is the Weil restriction of some $\Fq$-linear automorphism $a=b\circ\theta^j$. Therefore, if $W=\Psi_\alpha(V)$, then $\Am$ maps $W$ onto the Weil restriction of $b(V^{q^j})$.

Now let $\Am\in\Group$. We aim to show that there exists some integer $j$ such that $\Am\Thetam^{j}$ is the Weil restriction of an $\Fqm$-linear automorphism of $\Fqm^r$. Firstly, notice that since $\Am$ is a bijection as a linear map,
any linear Weil restriction can be seen as the preimage through $\Am$ of another linear Weil restriction, and as such is stabilized by both $\Jm_r$ and $\Am^{-1}\Jm_r\Am$. By Lemma \ref{lemma:stabilizers}, we conclude that $\Am^{-1}\Jm_r\Am$ is a polynomial in $\Jm_r$. Furthermore, conjugation by
the matrix $\Am^{-1}$ is an automorphism of the $\Fq$-algebra $\Fq^{rm\times rm}$, which means that it induces an $\Fq$-linear automorphism of $\Fq[\Jm_r]\simeq\Fqm$. In other words, conjugation by $\Am^{-1}$ defines an element of the Galois group $\Gal{\Fqm/\Fq}=\langle\theta\rangle$. This means that $\Am^{-1}\Jm_r\Am$ is
actually an $\Fq$-conjugate of $\Jm_r$, \textit{i.e.} of the form $\Jm_r^{q^j}$ for some integer $j$. Equivalently by Lemma \ref{lemma:conjugation}, there is an integer $j$ such that
\[\Am^{-1}\Jm_r\Am=\Thetam^{j}\Jm_r\Thetam^{-j},\]
which implies $\Am\Thetam^{j}$ commutes with $\Jm_r$. By Proposition \ref{prop:cond_r}, $\Am\Thetam^j=\Resq{\Bm}$ for some $\Bm\in\GL_r(\Fqm)$, which proves the theorem.
\end{proof}
\begin{corollary}
    Let $\Am\in\GL_{rm}(\Fq)$. Then $\Am\in\Group$ if, and only if there exists some integer $j$ such that $\Am\Jm_r\Am^{-1}=\Jm_r^{q^j}$.
\end{corollary}
\begin{proof}
Write $\Am=\Am_0\Thetam^j$ for some integer $j$. Since $\Am_0$ commutes with $\Jm_r$ by Proposition \ref{prop:cond_r}, we see that
$$\Am^{-1}\Jm_r\Am=\Thetam^{-j}\Am_0^{-1}\Jm_r\Am_0\Thetam^j=\Thetam^{-j}\Jm_r\Thetam^j=\Jm_r^{q^{m-j}},$$
the last equality coming from Lemma \ref{lemma:conjugation}. Raising the above equality to the power $q^j$ gives
$$\Am\Jm_r\Am^{-1}=\Jm_r^{q^j},$$
as required. The converse can be obtained in the same way by reading the above backwards.
\end{proof}
\begin{proposition}
    $\Group$ is a subgroup of $\GL_{rm}(\Fq)$, and its group structure is given by
    $$\Group\simeq\GL_r(\Fqm)\rtimes\mathbb{Z}/m\mathbb{Z}.$$
\end{proposition}
\begin{proof}
    We identify $\Group$ with the set of $\Fq$-linear automorphisms of $\Fqm^r$ of the form $b\circ\theta^j$ for $b\in\GL(\Fqm^r)$ and $0\leq j<m$. Then for any $c\in\GL(\Fqm^r)$ and $0\leq k<m$, we see that
    $$(c\circ\theta^k)\circ(b\circ\theta^j)=(c\circ b^{(q^k)})\circ \theta^{k+j},$$
    where $b^{(q^k)}$ is the $\Fqm$-automorphism whose matrix in the canonical basis is that of $b$ where all coefficients have been raised to the power $q^k$.
    This is indeed a semi-direct product group structure between $\GL(\Fqm^r)\simeq\GL_r(\Fqm)$ and $\langle\theta\rangle\simeq\mathbb{Z}/m\mathbb{Z}$.
\end{proof} \section{An attack against generic square-distinguishable alternant codes} \label{section:attack}
Let $\CC=\Alt{r}{\xv}{\yv}$ be a generic $q$-ary alternant code of extension degree $m$. We assume that we are in the square-distinguishable regime, \textit{i.e.} that the right-hand side of Inequality (\ref{eq:ea}) is greater than $n-\binom{rm+1}{2}$, so that 
Inequality (\ref{eq:ea}) is an equality by Heuristic \ref{heur:fgopt}. In the following, we denote by $\Hsec=\Psi_\alpha(\Vm_r(\xv,\yv))$ the secret parity-check matrix of $\CC$, where $\Vm_r(\xv,\yv)$ is the Vandermonde generator matrix of Remark \ref{remark:vandermonde}. The public key is another parity-check matrix 
$\Hpub=\Pm\Hsec$, where $\Pm$ is a secret $rm\times rm$ nonsingular $q$-ary matrix. We also denote by $\Vsec$ (resp. $\Vpub$) the geometric quadratic hull of $\Hsec$ (resp. $\Hpub$). By Proposition \ref{prop:intrinsic}, we have $\Vpub=\Pm\cdot\Vsec$. We aim to recover the private key, which is an efficient decoding algorithm for $\CC$. Our global approach follows a quite natural strategy consisting in 
retrieving some support $\xv'$ and multiplier $\yv'$.
\subsection{Case $r\leq q$} In the regime where $r\leq q$, Proposition \ref{prop:weil_proper_alt} ensures that $\CC$ is a Weil-proper alternant code. By Proposition \ref{prop:invariance}, \textit{all} the tangent spaces of $\Vpub$ are stabilized by $\Pm\Jm_r\Pm^{-1}$. These tangent spaces have dimension $2m$, as they are the image through $\Pm$ of the tangent spaces of $\Vsec$, themselves being the Weil restrictions of the 
tangent spaces of the quadratic hull of $\Vm_r(\xv,\yv)$ by Proposition \ref{prop:tangent}, which we know have dimension 2 as tangent spaces of the affine cone over a projective curve. Intuitively, there should not be much more matrices stabilizing all tangent spaces of $\Vsec$ than those that stabilize all Weil restrictions. Although this statement is likely to be provable, we state the following result as a heuristic. 
\begin{heur} \label{heur}
Generally, we have
$$\displaystyle\bigcap_{Q\in \Vpub}\St{T_Q \Vpub}=\Fq[\Pm \Jm_r \Pm^{-1}].$$
Experimentally, taking the intersection over only
$$N\eqdef \left\lceil\dfrac{1}{\rho(1-\rho)}\right\rceil$$
points, where $\rho=\dfrac{2}{r}$, suffices to get $\Fq[\Pm \Jm_r \Pm^{-1}]$.
\end{heur}
Let $\cA\eqdef\Fq[\Pm \Jm_r \Pm^{-1}]$. In order to compute $\cA$ in practice, one has to compute sufficiently many tangent spaces and then compute the space of matrices that stabilize all of them. Computing tangent spaces of $\Vpub$ can be done in polynomial time by merely computing the right kernel of the Jacobian matrix of a basis of the vanishing ideal at degree 2. Recall that computing rational points of $\Vpub$ is free, as we already have $n$ points 
given by the columns of $\Hpub$. Finally, computing the stabilizing algebra of all those tangent spaces can be done by simply solving a linear system of $2m\times (rm-2m)$ equations. Indeed, let $\Am=(a_{i,j})$ be an $rm\times rm$ matrix whose entries are unknowns. Let $T=T_{\gv} \Vpub$ be the tangent space of $\Vpub$ at some point $\gv$, that we may assume to be a column of $\Hpub$. Then $T$ can be seen as a linear code, and as such we can compute 
a generator matrix $\Gm$ and a parity-check matrix $\Hm$ of $T$. Then $\Am$ stabilizes $T$ if and only if for any $\word{h}\in T$, the vector $\Am\word{h}$ is orthogonal to the rows of $\Hm$. By linearity, the matrix $\Am$ stabilizes $T$ if and only if for any row $\word{r}$ of $\Gm$, we have $\Hm\Am\word{r}^\top=0$. All in all, the equations over the $a_{i,j}$'s expressing the fact that $\Am$ stabilizes $T$ are the coefficients of $\Hm\Am\Gm^\top$,
which is indeed an $(rm-2m)\times 2m$ matrix. Gathering these equations in a list for sufficiently many tangent spaces, we get a system of equations whose solution space is the intersection of the stabilizers of all these tangent spaces. For the system to be overdetermined, we need to compute at least 
$$\left\lceil\dfrac{(rm)^2}{2m(rm-2m)}\right\rceil=\left\lceil\dfrac{1}{\rho(1-\rho)}\right\rceil$$
tangent spaces, which explains Heuristic \ref{heur}. Step by step, we give a full algorithm for recovering $\cA$.
\begin{algorithm}[H]
    \caption{Computing $\cA=\Fq[\Pm \Jm_r \Pm^{-1}]$}
    \label{alg:cA}
    \begin{algorithmic}[1] \State \textbf{Input: } $\Hpub=(\word{g}_1|\ldots|\word{g}_n)$ the public parity-check matrix of $\CC$
    \State \textbf{Output: } $\Am_1,\ldots,\Am_m$ an $\Fq$-basis of $\cA$
    \State Compute a basis $\cF=(f_1,\ldots,f_N)$ of $I_2(\Hpub)$ \label{step:I2}
    \State Compute the Jacobian matrix $\mathbf{Jac}(\cF)=(\partial_{jk}f_i)_{i;j,k}\in\Sm_1^{N\times rm}$  \label{step:jac}
    \State $\Am\gets (a_{i,j})$ \Comment{Its entries are unknowns, or formal variables $a_{i,j}$}
    \State $\mathcal{S}\gets\varnothing$ \Comment{A set of linear equations over the $a_{i,j}$'s defining the stabilizers}
    \State $i\gets 0$
    \While{$\#\mathcal{S}<(rm)^2$} \Comment{While there are more unknowns than equations}
        \State Compute $T_i\eqdef T_{\word{g}_i}V$ as the right kernel of $\mathbf{Jac}(\cF)(\word{g}_i)$ \label{step:tangent}
        \State Compute a generator matrix $\Gm_i$ and a parity-check matrix $\Hm_i$ of $T_i$
        \State Add to $\mathcal{S}$ the coefficients of $\Hm_i \Am \Gm_i^{\top}$ \label{step:addeqs}
        \State $i\gets i+1$
    \EndWhile
    \State Compute $\Am_1,\ldots,\Am_m$ a basis of the solution space of $\mathcal{S}$
    \State \Return $\Am_1,\ldots,\Am_m$
    \end{algorithmic}
    \end{algorithm}

\begin{theorem} \label{thm:complexity1}
Assuming Heuristic \ref{heur} is true, Algorithm \ref{alg:cA} returns a basis of $\cA$ in $O\left(rn^\omega\right)$ operations in $\Fq$, where $\omega$ is the exponent of linear algebra.
\end{theorem}
\begin{proof}
Step \ref{step:I2} amounts to computing the left kernel of the matrix whose rows are $\word{r}_i\star\word{r}_j,~i\leq j$, where $\word{r}_i$ stands for the $i$-th row of $\Gpub$. This is a $\binom{rm+1}{2}\times n$ matrix. As we assume to be in a regime where 
there are no elements in $I_2(\Gpub)$ whose existence is forced by dimension, we have $\binom{rm+1}{2}=O(n)$ and the cost of Step \ref{step:I2} is therefore $O(n^\omega)$.\\
The cost of Step \ref{step:jac} is negligible compared to that of Step \ref{step:tangent}. Evaluating $\mathbf{Jac}(\cF)$ at $\word{g}_i$ is also negligible compared to computing the right-kernel of the resulting matrix, which costs $O(n^\omega)$ or even less since $\dim I_2(\Gpub)$
is typically significantly inferior to $n$. Step \ref{step:addeqs} adds $2m\times(rm-2m)$ equations to the set $\mathcal{S}$. The number of times we need to go through the loop to have $\#\mathcal{S}\geq (rm)^2$ is therefore 
$$\dfrac{(rm)^2}{2m(rm-2m)}=\dfrac{r^2}{2(r-2)}<r.$$ 
We conclude that we need to go through the loop $O(r)$ times to complete $\mathcal{S}$. Once this is done, Heuristic \ref{heur} ensures that the matrices that satisfy all equations of $\mathcal{S}$ are in $\cA$. Solving the system 
costs $O((rm)^{2\omega})=O(n^\omega)$ operations in $\Fq$. The overall complexity is therefore $O(rn^\omega)$, as we compute the right-kernel of the Jacobian matrix $O(r)$ times.
\end{proof}
\begin{remark}
Essentially, we need to compute sufficiently many equations so that the linear system $\mathcal{S}$ is overdetermined. Experimentally, having $\mathcal{S}$ overdetermined always suffices to have the solutions space reduced to $\cA$. If it is not the case,
we just need to add a few more equations by computing some additional tangent spaces. We expect the number of potential additional steps to be negligible, so that it does not change the total complexity of Algorithm \ref{alg:cA}.
\end{remark}

Using Algorithm \ref{alg:cA}, we get access to $\cA$. The next step is to use the field structure of $\cA$, which is ensured by the following.
\begin{proposition}
$\cA\simeq\Fqm$.
\end{proposition}
\begin{proof}
We already have $\Fq[\Jm_r]\simeq\Fq[\Jm]\simeq\Fqm$, the last isomorphism being $\Mata$. Now, define 
$$C_{\Pm}:\begin{cases}
    \Fq^{rm\times rm}&\longrightarrow\Fq^{rm\times rm}\\
    \Am &\longmapsto\Pm\Am\Pm^{-1},
\end{cases}$$
which we refer to as \textit{the conjugation map of} $\Pm$, and which is known to be an automorphism of the $\Fq$-algebra $\Fq^{rm\times rm}$. 
Therefore, $\cA=C_{\Pm}(\Fq[\Jm_r])$ inherits the field structure of $\Fq[\Jm_r]$, from which we conclude $\cA\simeq\Fqm$.
\end{proof}
Conjugation by the matrix $\Pm$ defines an isomorphism from $\Fq[\Jm_r]$ to $\cA$. Even if we know both $\cA$ and $\Jm_r$, we still do not have access to $\Pm$ nor its conjugation map, and not even to the restriction of the latter to $\Fq[\Jm_r]$.
What we can do, however, is draw $\Am\in\cA$ uniformly at random until we get a generator of $\cA$, \textit{i.e.} an element of degree $m$. There should not be many trials necessary to get such a matrix $\Am$. Indeed, it suffices to find a generator of the 
multiplicative group $\cA^\times$, which has the following proportion in $\cA$:
$$\pi=\ds\prod_p\left(1-\dfrac{1}{p}\right),$$
where the product is taken over all prime divisors $p$ of $q^m-1$. The probability to need more than $t$ trials to get a generator is thus $(1-\pi)^t$, which tends towards zero exponentially fast.

Let $\Pi_{\Am}$ be the minimal polynomial of such a matrix $\Am$ over $\Fq$. This polynomial is irreducible of degree $m$ over $\Fq$.
Since $\Fqm/\Fq$ is Galois, it splits into linear factors over $\Fqm$:
$$\Pi_{\Am}=\displaystyle\prod_{j=0}^{m-1}(X-\zeta^{q^j}),$$
where $\zeta\in\Fqm$ is therefore a primitive element. As such, there exists some polynomial $f\in\Fq[X]$ of degree at most $m-1$ such that $f(\zeta)=\alpha$. Applying $f$ on $\Am$, we get a matrix with the same minimal polynomial as $\Jm_r$,
which enables us to assume that $\Pi_{\Am}=\Pi_\alpha$ in the following.
\begin{lemma} \label{lemma:conjugate}
    There exists an integer $0\leq j < m$ such that $\Pm\Jm_r\Pm^{-1}=\Am^{q^j}$.
\end{lemma}
\begin{proof}
    The conjugation map $C_{\Pm}$ preserves minimal polynomials, therefore $\Pm\Jm_r\Pm^{-1}$ has the same minimal polynomial as $\Jm_r$, which is also that of $\Am$. As a result, $\Pm\Jm_r\Pm^{-1}$ is a Galois conjugate of 
    $\Am$, \textit{i.e.} of the form $\Am^{q^j}$ for some $j$.
\end{proof}
On the other hand, one can compute some matrix $\Qm$ satisfying $\Jm_r=\Qm\Am\Qm^{-1}$, the existence of which is proven below.
\begin{lemma} \label{lemma:similar}
    There exists $\Qm\in\GL_{rm}(\Fq)$ such that $\Jm_r=\Qm\Am\Qm^{-1}$.
\end{lemma}
\begin{proof}
    The two matrices share the same \textit{irreducible} minimal polynomial, which implies that their Jordan normal form is the same. They are therefore similar over $\Fqm$. As both of them have coefficients in $\Fq$, they are in fact 
    similar over $\Fq$.
\end{proof}
The matrix $\Qm$ of the previous lemma is exactly the one that we were looking for, as stated below.
\begin{theorem}
One can compute $\Qm\in\GL_{rm}(\Fq)$ such that $\Qm\Pm\in\Group$.
\end{theorem}
\begin{proof}
Let $\Qm$ be the matrix of Lemma \ref{lemma:similar}. From Lemma \ref{lemma:conjugate} and \ref{lemma:similar}, we get 
$$\begin{cases}
    \Pm\Jm_r\Pm^{-1}=\Am^{q^j}\\
    \Jm_r=\Qm\Am\Qm^{-1},
\end{cases}$$
from which we get $\Pm\Jm_r\Pm^{-1}=\Qm^{-1}\Jm_r^{q^j}\Qm$, or equivalently $(\Qm\Pm)\Jm_r (\Qm\Pm)^{-1}=\Jm_r^{q^j}$. By Theorem \ref{thm:group}, this means $\Qm\Pm\in\Group$.
\end{proof}
It only remains to explain why this solves our problem.
\begin{proposition} \label{prop:grs_cojugated}
    Let $\Gm'=\Psi_\alpha^{-1}(\Qm\Hpub)$. Then $\Gm'$ is a generator matrix of $\mathbf{GRS}(\xv^{q^j},\yv^{q^j})$.
\end{proposition}
\begin{proof}
    As $\Qm\Pm\in\Group$, it preserves Weil restrictions. More precisely, there is some $\Bm\in\GL_{r}(\Fqm)$ such that $\Qm\cdot\Vpub=\Qm\Pm\cdot\Vsec$ is the Weil restriction $\Bm\cdot\cY^{q^j}$, where $\cY$ is the geometric quadratic hull of $\Vm_r(\xv,\yv)$. This means that 
    the columns of $\Gm'$ are points of $\Bm\cdot\cY^{q^j}$. In other words, there exist polynomials $f_1,\ldots,f_r\in\Fqm[X]$ of degree at most $r-1$ such that 
    $$\Gm'=\begin{pmatrix}
        y_1^{q^j}f_1(x_1^{q^j}) & \ldots & y_n^{q^j}f_1(x_n^{q^j})\\
        \vdots & \ddots & \vdots\\
        y_1^{q^j}f_r(x_1^{q^j}) & \ldots & y_n^{q^j}f_r(x_n^{q^j})\\
    \end{pmatrix},$$
    which means that the row space of $\Gm'$ is a subspace of $\mathbf{GRS}(\xv^{q^j},\yv^{q^j})$. Since $\Qm\Hpub$ has rank $rm$, $\Gm'$ has rank $r$ and is therefore a generator matrix of $\mathbf{GRS}(\xv^{q^j},\yv^{q^j})$.
\end{proof}
The Sidelnikov-Shestakov attack then suffices to recover a support $\xv'$ and a multiplier $\yv'$ in $O(n^\omega)$ operations in $\Fqm$. The whole algorithm is detailed below.
\begin{algorithm}[H]
\caption{Recovering a support $\xv'$ and a multiplier $\yv'$ for $\CC$}
\label{alg:attack}
\begin{algorithmic}[1] \State \textbf{Input: } $\Hpub$ the public parity-check matrix of $\CC$
\State \textbf{Output: } $\xv',\yv'$ such that $\CC=\Alt{r}{\xv'}{\yv'}$.
\State Compute $\cA$ using Algorithm \ref{alg:cA} \label{step:cA}
\State Compute $\Am\in\cA$ such that $\Pi_{\Am}=\Pi_{\alpha}$ \label{step:Pi_A}
\State Compute $\Qm\in\GL_{rm}(\Fq)$ such that $\Jm_r=\Qm\Am\Qm^{-1}$
\State Compute $\Gm'=\Psi_{\word{\alpha}}^{-1}(\Qm\Hpub)$ and let $\DC=\{\word{m}\Gm'~|~\word{m}\in\Fqm^r\}$
\State Apply \cite{SS92} on $\DC$ to get $\xv',\yv'$ such that $\DC=\GRS{r}{\xv'}{\yv'}$ \label{step:SS92}
\State \textbf{Return} $\xv',\yv'$.
\end{algorithmic}
\end{algorithm} 

\begin{theorem}
Provided that $m=O(\log r)$, Algorithm \ref{alg:attack} returns a support $\xv'$ and a multiplier $\yv'$ such that $\CC=\Alt{r}{\xv'}{\yv'}$ at a cost of $O(rn^\omega)$ operations in $\Fq$.
\end{theorem}
\begin{proof}
Firstly, let us prove that $\xv',\yv'$ returned by Algorithm \ref{alg:attack} are indeed valid support and multiplier. By Proposition \ref{prop:grs_cojugated}, we have 
$$\GRS{r}{\xv'}{\yv'}=\mathbf{GRS}(\xv^{q^j},\yv^{q^j})$$
for some integer $0\leq j<m$. The GRS codes generated by $\xv,\yv$ and $\xv',\yv'$ respectively are therefore equal up to Galois conjugation, the latter being erased by the trace map. As a consequence, the two GRS codes have the same trace codes. By duality, 
$\xv,\yv$ and $\xv',\yv'$ generate the same alternant code.

Let us now derive the complexity of Algorithm \ref{alg:attack}. Step \ref{step:cA} requires $O(rn^\omega)$ operations in $\Fq$ by Theorem \ref{thm:complexity1}. Step \ref{step:Pi_A} requires computing matrices $\Am\in\cA$ at random $O(1)$ times, then computing the minimal polynomial and checking its degree.
Computing the minimal polynomial of $\Am$ requires $O((rm)^\omega)$ operations in $\Fq$, so Step \ref{step:Pi_A} is negligible. Computing the similarity matrix $\Qm$ requires computing the Jordan normal form of $\Am$ and checking whether
it is equal to that of $\Jm_r$ (which can be precomputed). This step requires $O((rm)^\omega)$ operations (possibly in $\Fqm$) and is again negligible. The only remaining costly operation is Step \ref{step:SS92} which is known to require $O(n^\omega)$
operations in $\Fqm$, which boils down to $O(mn^\omega)$ operations on $\Fq$. Since $m=O(\log r)$ --- which is the natural asymptotic regime in McEliece cryptosystem --- we conclude that the overall complexity is bounded by that of Algorithm \ref{alg:cA}.
\end{proof}
We provide a \texttt{SageMath} implementation of our attack, which is available \href{https://github.com/AxelfVGD/Tangent/tree/main}{here}.

\subsection{Case $r>q$} The attack described in the previous section requires the alternant code to be generic, square-distinguishable, and of degree $r\leq q$. If we run the exact same algorithm with a generic square-distinguishable alternant code of degree $r>q$
in input, then it turns out that the algorithm does return valid support $\xv'$ and multiplier $\yv'$. Although we will not provide a proof of this, we will try to give a partial explanation. To this end we introduce two results.
\begin{theorem}[Weil's theorem]
    Let $I\subset\Rm$ be an ideal and $J=\Resq{I}\subset\Sm$, where $\Rm$ and $\Sm$ are the polynomial rings defined in Section \ref{section:weil}. There exists an isomorphism of graded $\Fqm$-algebras
    \begin{equation}\label{eq:weil}
        (\Sm/J)\otimes_{\Fq}\Fqm\simeq\bigotimes_{j=0}^{m-1}\Rm/I^{q^j},
    \end{equation}
    where $I^{q^j}=\{f^{(q^j)}~|~f\in I\}$, the notation $f^{(q^j)}$ referring to $f$ where the automorphism $\theta^j$ has been applied on the coefficients --- the degree is preserved.
\end{theorem}
\begin{proof}
    This result is \cite[Theorem 2.8]{CCG23} in the case of finite fields.
\end{proof}
This means that we can work equivalently with the Weil restriction or the ideal obtained by extension of scalars. Note that this theorem somehow looks like the following result from coding theory.
\begin{proposition}[\cite{BMT23}]
    Let $\CC\subset\Fqm^n$ be a proper code. Then 
    $$\Tr{\CC}_{\Fqm}=\bigoplus_{j=0}^{m-1}\CC^{q^j},$$
    where $\Tr{\CC}_{\Fqm}$ denotes the $\Fqm$-linear code spanned by $\Tr{\CC}$.
\end{proposition}
To understand why Algorithm \ref{alg:attack} still works in the case $r>q$, we can therefore study the quadratic hull of 
$$\Alt{r}{\xv}{\yv}^\perp_{\Fqm}=\bigoplus_{j=0}^{m-1}\GRS{r}{\xv}{\yv}^{q^j}.$$
A natural choice for the generator matrix of this code is the matrix whose columns consist of the following list:
\begin{align*}
\cB =\{&\yv,\xv\yv,\ldots,\xv^{r-1}\yv\\
&\yv^q,(\xv\yv)^q,\ldots,(\xv^{r-1}\yv)^q\\
&\ldots\\
&\yv^{q^{m-1}},(\xv\yv)^{q^{m-1}},\ldots,(\xv^{r-1}\yv)^{q^{m-1}}\},
\end{align*}
which was first introduced in \cite{CMT23a} and referred to as \textit{the canonical basis}. For more convenience, we will analyze the quadratic hull of $\CC_{\Fqm}$ in the 
polynomial ring 
$$\Rm=\Fqm[X_{i,u}~|~0\leq i<r,~0\leq u<m],$$
the evaluation map sending $X_{i,u}$ onto $(\word{x}^i\star\word{y})^{q^u}$. As recalled in \cite{CMT23a}, the work that was done in \cite{FGOPT11} shows that, heuristically, the algebraic 
quadratic hull related to this basis $\cB$ is generated by the equations of the form 
$$f_{i,j,k,\ell,u,v}=X_{i,u}X_{j,v}-X_{k,u}X_{\ell,v},$$
for all $0\leq u,v<m$ and $0\leq i,j,k,\ell<r$ such that $iq^u+jq^v=kq^u+\ell q^v$. Denoting by $\cY$ the affine cone over the rational normal curve, we see that these equations define a subvariety of 
$$\cY\times\cY^q\times\ldots\times\cY^{q^{m-1}},$$
which is the variety obtained by extension of scalars of the quadratic hull of $\Alt{r}{\xv}{\yv}^\perp$ by Weil's theorem. We see that the defining ideal of this product variety is generated by the $f_{i,j,k,\ell,u,v}$'s with $u=v$.
Furthermore, it is the expected quadratic hull of $\Alt{r}{\xv}{\yv}^\perp_{\Fqm}$ when $r\leq q$, \textit{i.e.} when we have Weil-properness. On the other hand, the crossed equations $f_{i,j,k,\ell,u,v}$ where $u\neq v$ might be interpreted as field equations. We indeed see the following experimentally.
\begin{heur}\label{heur:points}
Let $\Hsec=\Psi_\alpha(\Vm_r(\xv,\yv))$. When $r>q$, the algebraic quadratic hull of $\Hsec$ is the Weil restriction of $\langle I_2(\Vm_r(\xv,\yv))\cup \{X_i^{q^m}-X_i~|~0\leq i<r\}\rangle$.
\end{heur}
The variety of which the algebraic quadratic hull of $\Hpub$ is the Weil restriction is thus the one-dimensional variety given by the lines passing through the points of the cone over the rational normal curve. The Weil restriction of a line is 
an $m$-dimensional vector-space, and therefore the arguments of Heuristic \ref{heur} still hold in this case. In fact, Weil-properness is not a necessary condition for the attack to succeed. What we need is to have a quadratic hull presenting some 
algebraic structure over $\Fqm$, so that we get access to sufficiently many vector spaces that are stabilized by $\Pm\Jm_r\Pm^{-1}$. All in all, our attack breaks McEliece scheme with generic square-distinguishable alternant codes, regardless of 
how $q$ compares to $r$. Our \texttt{SageMath} implementation generates such a generic alternant code, and recovers a support and a multiplier as soon as we are in the square-distinguishable regime. 

\subsection{What about Goppa codes ?} In both \cite{ABCCGLMMMNPPPSSSTW20} and \cite{M78}, McEliece cryptosystem is described with the family of binary Goppa codes.
\begin{definition}[Goppa code]
Let $\xv\in\Fqm^n$ be a support and $\Gamma\in\Fqm[X]$ be a polynomial of degree $r$ such that $\Gamma(x_i)\neq 0$ for all $i$. The Goppa code of support $\xv$ and Goppa polynomial $\Gamma$ is defined by
$$\Goppa{\xv}{\Gamma}=\Alt{r}{\xv}{\Gamma(\xv)^{-1}}.$$
\end{definition}
A result analogous to Theorem \ref{thm:ea} exists also for Goppa codes.
\begin{theorem}[\cite{MT21}] \label{thm:eg}
    Let $\CC=\Goppa{\xv}{\Gamma}^\perp$ be a proper dual Goppa code. If $r<q-1$, then 
    $$\dim I_2(\CC)\geq m\binom{r-1}{2}.$$
    When $r\geq q-1$, we have the following bound:
    \begin{equation} \label{eq:eg}
        \dim I_2(\CC)\geq \dfrac{m}{2}r\left((2\eg+1)r-2(q-1)q^{\eg-1}-1\right),
    \end{equation}
    where $\eg=\left\lceil\log_q\left(\dfrac{r}{(q-1)^2}\right)\right\rceil+1$.
\end{theorem}
Like in the alternant case, these bounds are tight in the following sense.
\begin{heur}[\cite{FGOPT11}]
    As soon as the right-hand side of the inequalities of Theorem \ref{thm:eg} exceed $\binom{rm+1}{2}-n$, these inequalities are equalities.
\end{heur}
As a direct consequence, we see that Goppa codes essentially behave like alternant codes when $r<q-1$, which means that Algorithm \ref{alg:attack} returns valid support and multiplier in such a case. When $r\geq q-1$ however, even the previous argument 
of Heuristic \ref{heur:points} may not apply. This is indeed what we see in practice.
\begin{heur}\label{heur:bad_news}
Let $\Hm$ be any parity-check matrix of a degree-$r$ Goppa code $\Goppa{\xv}{\Gamma}$. If $r\geq q-1$, then the dimension of the geometric quadratic hull of $\Hm$ is equal to 1.
\end{heur}
This is unfortunate, as no algebraic structure over $\Fqm$ is visible in the quadratic hull of such Goppa codes. In particular, binary Goppa codes, be they square-distinguishable or not, are again out of reach. We present in the following table a comparison between our 
attack and the current state of the art.
\begin{figure}[h]
    \begin{tabular}{|c|c|c|c|c|}
        \hline
        Target & Paper & $r (\geq 3)$ & $q$ & complexity\\
        \hline
        generic square dist. alternant & \cite{BMT23} & any & $q\in\{2,3\}$ & $rn^\omega$ \\
        \hline
        generic square dist. alternant & \cite{CMT23a} & $<q+1$ & any & $n^{\omega + 2}$\\
        \hline 
        generic square dist. alternant & \cite{CMT23a} + \cite{BMT23} & any & any & $n^{\omega + 2}$\\
        \hline 
        generic square dist. alternant & this paper & any & any & $rn^\omega$\\
        \hline
        square dist. Goppa & \cite{CMT23a} & $<q-1$ & any & $n^{\omega + 2}$\\
        \hline 
        square dist. Goppa & this paper & $<q-1$ & any & $rn^\omega$ \\
        \hline
    \end{tabular}
    \caption{Comparison}
\end{figure}

\subsection{Generalization to Algebraic-Geometry codes} Algebraic-Geometry (AG) codes, first introduced by Goppa in \cite{G81}, are a natural generalization of GRS codes. They have been proposed by Janwa and Moreno in \cite{JM96} for McEliece cryptosystem, but this proposition 
was eventually proven insecure in \cite{CMP14}. However, the security of McEliece with subfield subcodes of AG codes, often referred to as SSAG codes, remains unknown in general. In this last subsection, we aim to show that our algorithm can be applied to attack SSAG codes as well, 
provided that certain conditions are met.

We first need to recall some notions. In the following, $\cX$ denotes a smooth, projective and absolutely irreducible algebraic curve defined over $\Fqm$. We denote by $g$ the genus of $\cX$. We also denote by $\Fqm(\cX)$ the function field of $\cX$ with field of constants $\Fqm$. Recall that
$\mathrm{Div}(\cX)$ stands for the divisor group of $\cX$, \textit{i.e.} the free abelian group generated by the points of $\cX$ over the algebraic closure of $\Fqm$. A divisor $D\in\mathrm{Div}(\cX)$ therefore takes the form 
$$D=\ds\sum_{P\in\cX}n_P\cdot (P),$$
where all but finitely many $n_P$'s are zero. We write 
$$\mathrm{supp}(D)=\{P\in\cX~|~n_P\neq 0\},$$
and call it the support of $D$. The degree of $D$ is defined by $\deg D=\sum_{P\in\cX}n_P\in\mathbb{Z}$. We say that $D$ is an \textit{effective} divisor, and write $D\geq 0$, when $\forall P\in\cX,~n_P\geq 0$. The divisor $D$ is said to be defined over $\Fqm$ when 
$$D^{q^m}\eqdef\ds\sum_{P\in\cX}n_P(P^{q^m})=D,$$
where $P^{q^m}$ is the point of $\cX$ with all coordinates being those of $P$ to the power $q^m$ --- this is indeed a point of $\cX$ as $\cX$ is defined over $\Fqm$. The subgroup of divisors defined over $\Fqm$ will be referred to as $\mathrm{Div}_{\Fqm}(\cX)$. Note that the points of a divisor defined over $\Fqm$ need not have their coordinates in $\Fqm$.
This remark is particularly relevant when we look at principal divisors, which are divisors associated to nonzero rational functions in the following manner. For $f\in\Fqm(\cX)^\times$ such a nonzero rational function and $P\in\cX$, we write $\mathrm{ord}_P(f)$ the valuation of $f$ at $P$. The principal divisor associated to $f$ is defined by 
$$(f)=\ds\sum_{P\in\cX}\mathrm{ord}_P(f)\cdot(P).$$
Although the zeros and poles of $f$ may not lie in $\cX(\Fqm)$, the whole divisor $(f)$ is globally invariant under the Frobenius automorphism and is therefore defined over $\Fqm$. Furthermore, the following holds.
\begin{theorem}[\cite{S09a}, Proposition 3.1 (b)]
For all $f\in\Fqm(\cX)^\times$, we have $\deg (f)=0$.
\end{theorem}
Finally, we denote by $\cL(D)$ the Riemann-Roch space of $D$, which is defined by 
$$\cL(D)=\{f\in\Fqm(\cX)^\times~|~(f)+D\geq 0\}\cup\{0\}.$$
As stated by the famous Riemann-Roch theorem, this space is a finite dimensional $\Fqm$-vector space whose dimension $\ell(D)$ is related to the degree of $D$. We give the special case of Riemann-Roch theorem for divisors 
of sufficiently high degree.
\begin{theorem}[Riemann-Roch]
    If $\deg D > 2g-2$, then $\ell(D)=\deg D + 1 - g$.
\end{theorem}
\begin{proof}
    See \cite[Corollary 5.5 (c)]{S09a}.
\end{proof}
We have introduced all the necessary material to define algebraic-geometry codes. 
\begin{definition}
Let $\cP=(P_1,\ldots,P_n)\in\cX(\Fqm)^n$ be a tuple of pairwise distinct $\Fqm$-rational points of $\cX$, and let $D\in\mathrm{Div}_{\Fqm}(\cX)$
be such that $\mathrm{supp}(D)\cap \{P_1,\ldots,P_n\}=\varnothing$. The algebraic-geometry code of support $\cP$ and divisor $D$ is defined by 
$$\CC_{\cL}(\cX,\cP,D)=\{(f(P_1),\ldots,f(P_n))~|~f\in\cL(D)\}.$$
\end{definition}
\begin{remark}
    When $\cX=\PP^1$ and $D=(r+1)P_\infty$ where $P_\infty=(1:0)$, the above definition gives Reed-Solomon codes. See \cite[Proposition 2.3.5]{S09} for further details.
\end{remark}
Let $r=\ell(D)-1$ and $(f_0,\ldots,f_r)$ be a basis of $\cL(D)$. Since $\cX$ is a smooth curve, the rational function $\phi_D=(f_0,\ldots,f_r)$ defines a morphism $\cX\rightarrow\PP^r$. If $\deg D > 2g$, then 
$\phi_D$ defines an isomorphism between $\cX$ and $\cY=\phi_D(\cX)$. Moreover, like in the case of GRS codes, the following matrix 
$$V(\cP,\phi_D)=\begin{pmatrix}
    f_0(P_1) & f_0(P_2) & \ldots & f_0(P_n)\\
    f_1(P_1) & f_1(P_2) & \ldots & f_1(P_n)\\
    \vdots & \vdots & \ddots & \vdots\\
    f_r(P_1) & f_r(P_2) & \ldots & f_r(P_n)
\end{pmatrix}$$
is a generator matrix of $\CC\eqdef\CC_{\cL}(\cX,\cP,D)$, and each column defines a point on $\cY$. From the work of \cite{CMP11a}, we know that if $\deg D>2g+1$, then the quadratic hull of this matrix is equal to $\cY$, which 
implies that the quadratic hull of $\Tr{\CC}$ is a subvariety of $\Psi_\alpha(\cY)$. In order to adapt our attack to SSAG codes, we need to understand when such a code is Weil-proper. Extensive computations in our \texttt{SageMath}
implementation available \href{https://github.com/AxelfVGD/Tangent/tree/main}{here} have led us to the following quite natural conjecture. We have investigated Weil-properness of \textit{generic} one point AG codes, \textit{i.e.}
codes of the form
\begin{equation}\label{eq:one_point_ag}
\CC=\{\word{c}\star\yv~|~\word{c}\in\CC_{\cL}(\cX,\cP,r\cdot (P_\infty))\}
\end{equation}
where $\yv\in(\Fqm^\times)^n$ is chosen uniformly at random.
\begin{heur}[Weil-properness of one point generic AG codes]
Let $\DC=\Tr{\CC}$ where $\CC$ is as in Equation (\ref{eq:one_point_ag}). We make the two following assumptions:
\begin{itemize}
    \item[$(i)$] $2g+2\leq r \leq q$;
    \item[$(ii)$] $n>\binom{(r+1-g)m+1}{2}-m\dim_{\Fqm} I_2(\CC)$. 
\end{itemize}
Then $\DC$ is Weil-proper.
\end{heur}
In such a case, one can compute the tangent spaces of the quadratic hull of the public generator matrix of $\DC$. The assumption $2g+2\leq r$ ensures that the quadratic hull of $\CC$ is a projective variety of dimension 1, therefore the 
underlying ideal has dimension 2, which means that the quadratic hull of $\DC$ has dimension $2m$. All in all, Algorithm \ref{alg:attack} eventually retrieves a generator matrix of $\CC$, and then \cite{CMP14} recovers an efficient decoding
algorithm. This class of AG codes is therefore vulnerable to our attack.

Like in the case of alternant codes, the attack still succeeds when we no longer have the assumption $r\leq q$.
\begin{heur}
Let $\Hsec=\Psi_\alpha(\Gm)$ be the secret generator matrix of $\DC$, where $\Gm$ is a generator matrix of $\CC$. We still assume $2g+2\leq r$, as well as $(ii)$ of the previous heuristic, but now we also assume $r>q$. Then the algebraic quadratic hull 
of $\Hsec$ is equal to the Weil restriction of $\langle I_2(\Gm)\cup \{X_i^{q^m}-X_i~|~0\leq i \leq r-g\}\rangle$.
\end{heur}
Our \texttt{SageMath} implementation also provides Algorithm \ref{alg:attack} in the case of generic one point AG codes, regardless of how $r$ compares to $q$.
 \section{Conclusion and open problems}

The geometric analysis of linear codes using the quadratic hull seems to be a prolific approach. Using Weil restriction,
we have been able to write an algorithm that recovers the structure of a trace code, provided that the original code 
has a nontrivial quadratic hull and assuming Weil-properness. This notion turns out to be more powerful than what we 
actually need for our attack to succeed. We eventually described a polynomial-time attack against the McEliece cryptosystem
instantiated with generic alternant codes, generic one-point SSAG codes, or even Goppa codes of sufficiently low degree. 

For future work, better understanding what happens when the degree of the alternant code gives rise to field equations even in 
the high-rate regime will be crucial to see whether we can adapt our framework to the case of binary Goppa codes. It is also
still unclear whether other families of codes are vulnerable to our attack. \section*{Acknowledgements}
This work was in part funded by the \textit{Direction Générale de l'Armement} (DGA). The author would also like to thank Jean-Pierre Tillich for his 
advice and feedback, as well as Alain Couvreur for his insightful discussions.
\newpage 
\bibliographystyle{alpha}

\begin{thebibliography}{CGG{\etalchar{+}}13}

\bibitem[AAB{\etalchar{+}}17]{HQC17}
Carlos {Aguilar Melchor}, Nicolas Aragon, Slim Bettaieb, Lo\"{i}c Bidoux,
  Olivier Blazy, Jean{-}Christophe Deneuville, Philippe Gaborit, Edoardo
  Persichetti, and Gilles Z{\'{e}}mor.
\newblock {HQC}, November 2017.
\newblock NIST Round 1 submission for Post-Quantum Cryptography.

\bibitem[AAB{\etalchar{+}}22]{AABBBBDGGGGMPRSTVZ22}
Carlos {Aguilar Melchor}, Nicolas Aragon, Paulo Barreto, Slim Bettaieb,
  Lo{\"i}c Bidoux, Olivier Blazy, Jean-Christophe Deneuville, Philippe Gaborit,
  Santosh Ghosh, Shay Gueron, Tim G{\"u}neysu, Rafael Misoczki, Edoardo
  Persichetti, Jan Richter-Brockmann, Nicolas Sendrier, Jean-Pierre Tillich,
  Valentin Vasseur, and Gilles Z{\'e}mor.
\newblock {BIKE}.
\newblock Round 4 Submission to the NIST Post-Quantum Cryptography Call,
  v.~5.1, October 2022.

\bibitem[ABC{\etalchar{+}}22]{ABCCGLMMMNPPPSSSTW20}
Martin Albrecht, Daniel~J. Bernstein, Tung Chou, Carlos Cid, Jan Gilcher, Tanja
  Lange, Varun Maram, Ingo von Maurich, Rafael Mizoczki, Ruben Niederhagen,
  Edoardo Persichetti, Kenneth Paterson, Christiane Peters, Peter Schwabe,
  Nicolas Sendrier, Jakub Szefer, Cen~Jung Tjhai, Martin Tomlinson, and Wang
  Wen.
\newblock Classic {M}c{E}liece (merger of {Classic McEliece} and {NTS-KEM}).
\newblock \url{https://classic.mceliece.org}, November 2022.
\newblock Fourth round finalist of the NIST post-quantum cryptography call.

\bibitem[BM17]{BM17}
Leif Both and Alexander May.
\newblock {Optimizing {BJMM} with Nearest Neighbors: Full Decoding in $2^{2/21
  n}$ and {McEliece} Security}.
\newblock In {\em WCC Workshop on Coding and Cryptography}, September 2017.

\bibitem[BMT24]{BMT23}
Magali Bardet, Rocco Mora, and Jean-Pierre Tillich.
\newblock Polynomial time key-recovery attack on high rate random alternant
  codes.
\newblock {\em IEEE Trans. Inform. Theory}, 70(6):4492--4511, 2024.

\bibitem[BMvT78]{BMT78}
Elwyn Berlekamp, Robert McEliece, and Henk van Tilborg.
\newblock On the inherent intractability of certain coding problems.
\newblock {\em IEEE Trans. Inform. Theory}, 24(3):384--386, May 1978.

\bibitem[CCG23]{CCG23}
Alessio Caminata, Michela Ceria, and Elisa Gorla.
\newblock The complexity of solving weil restriction systems.
\newblock {\em Journal of Algebra}, 621:116--133, 2023.

\bibitem[CCMZ15]{CCMZ15}
Igniacio Cascudo, Ronald Cramer, Diego Mirandola, and Gilles Z{\'e}mor.
\newblock Squares of random linear codes.
\newblock {\em IEEE Trans. Inform. Theory}, 61(3):1159--1173, 3 2015.

\bibitem[CFS01]{CFS01}
Nicolas Courtois, Matthieu Finiasz, and Nicolas Sendrier.
\newblock How to achieve a {McEliece}-based digital signature scheme.
\newblock In {\em Advances in Cryptology - ASIACRYPT~2001}, volume 2248 of {\em
  LNCS}, pages 157--174, Gold Coast, Australia, 2001. Springer.

\bibitem[CGG{\etalchar{+}}13]{CGGOT13}
Alain Couvreur, Philippe Gaborit, Val{\'e}rie Gautier, Ayoub Otmani, and
  Jean-Pierre Tillich.
\newblock {Distinguisher-Based Attacks on Public-Key Cryptosystems Using
  {R}eed-{S}olomon Codes}.
\newblock In {\em {International Workshop on Coding and Cryptography - WCC
  2013}}, pages 181--193, Bergen, Norway, April 2013.

\bibitem[CMCP14]{CMP14}
Alain Couvreur, Irene M{\'a}rquez-Corbella, and Ruud Pellikaan.
\newblock A polynomial time attack against algebraic geometry code based public
  key cryptosystems.
\newblock In {\em Proc. IEEE Int. Symposium Inf. Theory - ISIT~2014}, pages
  1446--1450, June 2014.

\bibitem[CMT23]{CMT23a}
Alain Couvreur, Rocco Mora, and Jean-Pierre Tillich.
\newblock A new approach based on quadratic forms to attack the {McEliece}
  cryptosystem.
\newblock In Jian Guo and Ron Steinfeld, editors, {\em Advances in Cryptology -
  {ASIACRYPT} 2023 - 29th International Conference on the Theory and
  Application of Cryptology and Information Security, Guangzhou, China,
  December 4-8, 2023, Proceedings, Part {IV}}, volume 14441 of {\em LNCS},
  pages 3--38. Springer, 2023.

\bibitem[Del75]{D75}
Philippe Delsarte.
\newblock On subfield subcodes of modified {Reed-Solomon} codes.
\newblock {\em IEEE Trans. Inform. Theory}, 21(5):575--576, 1975.

\bibitem[Eis06]{eisenbud2006geometry}
D.~Eisenbud.
\newblock {\em The Geometry of Syzygies: A Second Course in Algebraic Geometry
  and Commutative Algebra}.
\newblock Graduate Texts in Mathematics. Springer New York, 2006.

\bibitem[FGO{\etalchar{+}}10]{FGOPT10}
Jean-Charles Faug{\`e}re, Val{\'e}rie Gauthier, Ayoub Otmani, Ludovic Perret,
  and Jean-Pierre Tillich.
\newblock A distinguisher for high rate {McEliece} cryptosystems.
\newblock IACR Cryptology ePrint Archive, Report2010/331, 2010.
\newblock {http://eprint.iacr.org/}.

\bibitem[FGO{\etalchar{+}}11]{FGOPT11}
Jean-Charles Faug{\`e}re, Val{\'e}rie Gauthier, Ayoub Otmani, Ludovic Perret,
  and Jean-Pierre Tillich.
\newblock A distinguisher for high rate {McEliece} cryptosystems.
\newblock In {\em Proc. IEEE Inf. Theory Workshop- ITW~2011}, pages 282--286,
  Paraty, Brasil, October 2011.

\bibitem[GH78]{GH78}
Phillip Griffiths and Joseph Harris.
\newblock {\em Principles of algebraic geometry}.
\newblock Pure and Applied Mathematics. A Wiley-Interscience Publication. John
  Wiley {\&} Sons, New York, 1978.

\bibitem[Gop81]{G81}
Valerii~D. Goppa.
\newblock Codes on algebraic curves.
\newblock {\em Dokl. Akad. Nauk SSSR}, 259(6):1289--1290, 1981.
\newblock In Russian.

\bibitem[JM96]{JM96}
Heeralal Janwa and Oscar Moreno.
\newblock {McEliece} public key cryptosystems using algebraic-geometric codes.
\newblock {\em Des. Codes Cryptogr.}, 8(3):293--307, 1996.

\bibitem[KT17]{KT17a}
Ghazal Kachigar and Jean-Pierre Tillich.
\newblock Quantum information set decoding algorithms.
\newblock In {\em Post-Quantum Cryptography~2017}, volume 10346 of {\em LNCS},
  pages 69--89, Utrecht, The Netherlands, June 2017. Springer.

\bibitem[LMT25]{LMT25}
Axel Lemoine, Rocco Mora, and Jean-Pierre Tillich.
\newblock Understanding the new distinguisher of alternant codes at degree 2.
\newblock Cryptology {ePrint} Archive, Paper 2025/531, 2025.

\bibitem[McE78]{M78}
Robert~J. McEliece.
\newblock {\em A Public-Key System Based on Algebraic Coding Theory}, pages
  114--116.
\newblock Jet Propulsion Lab, 1978.
\newblock DSN Progress Report 44.

\bibitem[MMP14]{CMP11a}
Irene {M{\'a}rquez--Corbella}, Edgar {Mart{\'\i}nez--Moro}, and Ruud Pellikaan.
\newblock On the unique representation of very strong algebraic geometry codes.
\newblock {\em Des. Codes Cryptogr.}, 70(1--2):215--230, 2014.

\bibitem[MP12]{MP12}
Irene {M{\'a}rquez-Corbella} and Ruud Pellikaan.
\newblock Error-correcting pairs for a public-key cryptosystem.
\newblock CBC 2012, Code-based {C}ryptography {W}orkshop, 2012.
\newblock Available on \url{http://www.win.tue.nl/~ruudp/paper/59.pdf}.

\bibitem[MT21]{MT21}
Rocco Mora and Jean-Pierre Tillich.
\newblock On the dimension and structure of the square of the dual of a {Goppa}
  code.
\newblock preprint, 2021.

\bibitem[Nie86]{N86}
Harald Niederreiter.
\newblock Knapsack-type cryptosystems and algebraic coding theory.
\newblock {\em Problems of Control and Information Theory}, 15(2):159--166,
  1986.

\bibitem[Ran19]{R20}
Hugues Randriambololona.
\newblock The quadratic hull of a code and the geometric view on multiplication
  algorithms.
\newblock {\em CoRR}, abs/1912.06627, 2019.

\bibitem[Ran24]{R24}
Hugues Randriambololona.
\newblock The syzygy distinguisher.
\newblock {\em CoRR}, abs/2407.15740, 2024.

\bibitem[Sil09]{S09a}
Joseph~H. Silverman.
\newblock {\em The Arithmetic of Elliptic Curves}, volume 106 of {\em Graduate
  Texts in Mathematics}.
\newblock Springer, 2009.

\bibitem[SS92]{SS92}
Vladimir~Michilovich Sidelnikov and S.O. Shestakov.
\newblock On the insecurity of cryptosystems based on generalized
  {Reed-Solomon} codes.
\newblock {\em Discrete Math. Appl.}, 1(4):439--444, 1992.

\bibitem[Sti09]{S09}
Henning Stichtenoth.
\newblock {\em Algebraic function fields and codes}, volume 254 of {\em
  Graduate Texts in Mathematics}.
\newblock Springer-Verlag, Berlin, second edition, 2009.

\bibitem[WB86]{BW86}
Lloyd~R Welch and Elwyn~R Berlekamp.
\newblock Error correction for algebraic block codes, December 1986.
\newblock US Patent 4,633,470.

\end{thebibliography}
\newcommand{\etalchar}[1]{$^{#1}$}

\newpage
\section{Appendix}
\subsection{Graded free resolutions}
The proof of Lemma \ref{lemma:genset} involves the notion of minimal graded free resolution of finitely generated graded modules over polynomial rings. Let us define 
\begin{center}
    $\Rm=\Fqm[X_0,\ldots,X_{r-1}]$ and $\Sm=\Fq[x_{i,j}~|~0\leq i<r,~0\leq j<m]$.
\end{center}
Both $\Rm$ and $\Sm$ have a structure of graded rings. More precisely, we can write 
\begin{center}
$\Rm=\bigoplus_{d\in\mathbb{N}}\Rm_d$ and $\Sm=\bigoplus_{d\in\mathbb{N}}\Sm_d$,
\end{center}
where $\Rm_d$ and $\Sm_d$ refer to the vector space spanned by monomials of degree $d$, and we have $\Rm_a\Rm_b\subset\Rm_{a+b}$ for all $a,b\in\mathbb{N}$ --- and likewise for $\Sm$. In the same manner, if $M$ is a module over $\Rm$, we say that $M$ is a graded $\Rm$-module if one can write $M=\bigoplus_{d\in\mathbb{Z}}M_d$ and if 
\[\forall a\in\mathbb{N},~\forall b\in\mathbb{Z},~\Rm_a M_b\subset M_{a+b}.\]

\begin{definition}
    Let $M$ be a finitely generated $\Rm$-module. A free resolution of $M$ is an exact sequence 
    $$\ldots\rightarrow\F_{i+1}\rightarrow\F_i\rightarrow\ldots\rightarrow\F_0\rightarrow M\rightarrow 0,$$ 
    where each $\F_i$ is a free $\Rm$-module. We write $\F_\bullet\rightarrow M \rightarrow 0$ for conciseness.
    If moreover each $\F_i$ is a graded free $\Rm$-module, and if the maps $\F_{i+1}\rightarrow\F_i$ and $\F_0\rightarrow M$ preserve the degree,
    then we say that $\F_\bullet\rightarrow M\rightarrow 0$ is a graded free resolution of $M$.
\end{definition}
\subsection{Minimal resolutions} Let $M$ be a finitely generated $\Rm$-module. If $(x_1,\ldots,x_n)\in M^n$ is a minimal set of generators of $M$, then we have a natural map 
$$\F_0=\Rm^n\twoheadrightarrow M.$$
One can repeat this process inductively by setting $M_1=\ker(\F_0\rightarrow M)$ and let $\F_1$ be the free module over a minimal set of generators of $M_1$, and so on. This is the notion of minimality of a free resolution, that 
we define more formally below.
\begin{definition}
Let $\F_\bullet\rightarrow M\rightarrow 0$ be a free resolution of the finitely generated graded $\Rm$-module $M$. Define the sequence of $\Rm$-modules $(M_i)_{i\in\mathbb{N}}$ by $M_0=M$ and $M_{i+1}=\ker(F_i\rightarrow M_i)$.
We say that the resolution is minimal if for any $i\in\mathbb{N}$, the cardinality of a minimal set of generators of $M_i$ equals the rank of the free module $\F_i$.
\end{definition}
If $N$ is any graded $\Rm$-module, then for any integer $j$ we denote by $N(j)$ the graded $\Rm$-module whose degree $d$ component is defined by the degree $d+j$ component of $N$:
$$\forall d\in\mathbb{Z},~N(j)_d=N_{d+j}.$$
We now assume that $M$ is a graded $\Rm$-module. By taking into account the degree of the elements of a minimal set of generators of $M$, we get a degree-preserving natural surjective map 
$$\F_0=\bigoplus_{j\in\mathbb{N}}\Rm(-j)^{\beta_{0,j}}\twoheadrightarrow M,$$
where $\beta_{0,j}$ is therefore the number of elements of degree $j$ in a minimal set of generators of $M$. Repeating this construction inductively, we obtain a minimal graded free resolution of $M$. The fact that the 
numbers of generators of a certain degree in a minimal set of generators does not depend on the choice of generators is a consequence of the following result.
\begin{theorem}[\cite{eisenbud2006geometry}, Theorem 1.6] 
    If $\F_\bullet$ and $\mathbb{G}_\bullet$ are minimal graded free resolutions of $M$, then there is a graded isomorphism of complexes 
    $\F_\bullet\rightarrow\mathbb{G}_\bullet$ inducing the identity map on $M$.
\end{theorem}
We can therefore talk about \textit{the} minimal graded free resolution of the finitely generated graded $\Rm$-module $M$. For any integer $j\in\mathbb{N}$, we can write 
$$\F_i=\bigoplus_{j\in\mathbb{N}}\Rm(-j)^{\beta_{i,j}},$$
and the $\beta_{i,j}\in\mathbb{N}$ are called the \textit{graded Betti numbers} of $M$.
\subsection{Minimal set of generators of a Weil restriction} A homogeneous ideal $I\subset\Rm$ is a graded $\Rm$-module. Since $\Rm$ is a noetherian ring by Hilbert's basis theorem, we know that $I$ is finitely generated.
As a result we can apply the framework that we introduced above on the finitely generated graded module $\Rm/I$. Furthermore, if $J=\Resq{I}$, then $J$ is also a finitely generated graded $\Sm$-module, and as such we can also consider 
the minimal graded free resolution of $\Sm/J$. The Betti numbers of $\Rm/I$ and $\Sm/J$ are related by the following formula.
\begin{proposition}\label{prop:betti}
Let $(\beta_{i,j})_{i,j}$ and $(\gamma_{i,j})_{i,j}$ be the graded Betti numbers of $\Rm/I$ and $\Sm/J$ respectively. Then for all $i$ and $j$, we have 
\begin{equation}\label{eq:betti}
    \gamma_{i,j}=\ds\sum_{\substack{i_1+\ldots+i_m=i \\ j_1+\ldots+j_m=j}}\prod_{s=1}^m\beta_{i_s,j_s}.
\end{equation}
\end{proposition}
\begin{proof}
    Let $\F_\bullet$ be a minimal graded free resolution of $\Rm/I$. By \cite[Theorem 3.3]{CCG23}, we have a minimal graded free resolution 
    \begin{equation} \label{eq:resolution}
    \mathbb{G}_\bullet\eqdef\F_\bullet^{\otimes m}\longrightarrow(\Sm/J)\otimes_{\Fq} \Fqm\longrightarrow 0.
    \end{equation}
    We thus have for all $i\in\mathbb{N}$,
    \[\mathbb{G}_i=\ds\bigoplus_{i_1+\ldots+i_m=i}\F_{i_1}\otimes\ldots\otimes\F_{i_m}.\]
    Expanding each $\F_{i_s}$ as $\bigoplus_{j_s}\Rm(-j_s)^{\beta_{i_s,j_s}}$ gives
    \begin{align*}
        \mathbb{G}_i&=\ds\bigoplus_{i_1+\ldots+i_m=i} \bigotimes_{s=1}^m\left(\bigoplus_{j_s}\Rm(-j_s)^{\beta_{i_s,j_s}}\right)\\
        &=\ds\bigoplus_{i_1+\ldots+i_m=i}\bigoplus_{j_1,\ldots,j_m} \Rm(-j_1)^{\beta_{i_1,j_1}}\otimes\ldots\otimes\Rm(-j_m)^{\beta_{i_m,j_m}}\\
        &=\ds\bigoplus_{i_1+\ldots+i_m=i}\bigoplus_{j_1,\ldots,j_m} \Rm(-j_1-\ldots-j_m)^{\beta_{i_1,j_1}\times\ldots\times\beta_{i_m,j_m}}.
    \end{align*}
    Since the map $\Phi:\Rm\rightarrow\Sm\otimes_{\Fq}\Fqm$ is an injective homomorphism of graded $\Fqm$-algebras, we get that the Betti numbers of $\Sm/J$ and $(\Sm/J)\otimes_{\Fq}\Fqm$ are the same.
    Identifying the component of degree $j$ of $\mathbb{G}_i$ eventually gives the desired formula.
\end{proof}
Focusing on $i=1$ gives the number of generators in a minimal set of generators of the ideals $I$ and $J$ respectively. The leads us to a proof of Lemma \ref{lemma:genset}, that we recall below.
\genset*
\begin{proof}
    
    Retaking the notations of Proposition \ref{prop:betti}, it suffices to show that
    $$\forall j\in\mathbb{N},~\gamma_{1,j}=m\beta_{1,j}.$$
    Let $j\in\mathbb{N}$. In Equation (\ref{eq:betti}), every index $i_s$ must be equal to zero but one, which must be equal to one. As $\F_0=\Rm$, we see that for each integer $j_s$, the Betti number $\beta_{0,j_s}$ equals 1 if and only if $j_s=0$, and $0$ otherwise.
    Therefore, in the product of Equation (\ref{eq:betti}), we only have Betti numbers of the form $\beta_{0,0}$ and one of the form $\beta_{1,j_s}$, and that $j_s$ must be equal to $j$. All in all, we get 
    $$\gamma_{1,j}=\ds\sum_{s=1}^m \beta_{0,0}\times\ldots\times\underbrace{\beta_{1,j}}_{\text{position }s}\times\ldots\times\beta_{0,0}=m\beta_{1,j}$$
    as required.
\end{proof}
 \end{document}